\documentclass{llncs}
\usepackage[utf8]{inputenc}
\usepackage[T1]{fontenc}
\usepackage{lmodern}
\usepackage{standalone}
\usepackage{textgreek}
\usepackage{microtype}
\usepackage[bookmarks,bookmarksopen,bookmarksdepth=2,colorlinks,unicode]{hyperref}
\usepackage{xspace}
\usepackage{float}
\usepackage{appendix}
\usepackage{pgfplots}\pgfplotsset{compat=1.18}
\usepackage{amsmath}
\usepackage{amsfonts}
\usepackage{amssymb}
\usepackage{braket}
\usepackage{mathtools}
\usepackage{siunitx}
\usepackage{enumitem}
\usepackage{algorithm}
\usepackage[noend]{algpseudocode}
\usepackage{array}
\usepackage{booktabs}
\usepackage{cleveref}
\usepackage{cite}
\usepackage{comment}
\usepackage{doi}
\usepackage{color}

\usetikzlibrary{cd}

\setlist[enumerate,1]{label={(\arabic*)}}

\newcommand\StateX{\Statex\hspace{\algorithmicindent}}
\makeatletter
\newcounter{algorithmicH}
\let\oldalgorithmic\algorithmic
\renewcommand{\algorithmic}{%
    \stepcounter{algorithmicH}
    \oldalgorithmic}
\renewcommand{\theHALG@line}{ALG@line.\thealgorithmicH.\arabic{ALG@line}}
\makeatother

\allowdisplaybreaks

\spnewtheorem{thm}[theorem]{Theorem}{\bfseries}{\itshape}
\spnewtheorem{defn}[theorem]{Definition}{\bfseries}{\itshape}
\spnewtheorem{prop}[theorem]{Proposition}{\bfseries}{\itshape}
\spnewtheorem{lem}[theorem]{Lemma}{\bfseries}{\itshape}
\spnewtheorem{cor}[theorem]{Corollary}{\bfseries}{\itshape}
\spnewtheorem{hyp}[theorem]{Hypothesis}{\bfseries}{\itshape}
\spnewtheorem{ex}[theorem]{Example}{\bfseries}{}
\spnewtheorem{rem}[theorem]{Remark}{\bfseries}{}

\crefname{thm}{theorem}{theorems}
\Crefname{thm}{Theorem}{Theorems}
\crefname{defn}{definition}{definitions}
\Crefname{defn}{Definition}{Definitions}
\crefname{prop}{proposition}{propositions}
\Crefname{prop}{Proposition}{Propositions}
\crefname{lem}{lemma}{lemmas}
\Crefname{lem}{Lemma}{Lemmas}
\crefname{cor}{corollary}{corollaries}
\Crefname{cor}{Corollary}{Corollaries}
\crefname{ex}{example}{examples}
\Crefname{ex}{Example}{Examples}
\crefname{rem}{remark}{remarks}
\Crefname{rem}{Remark}{Remarks}
\crefname{hyp}{hypothesis}{hypotheses}
\Crefname{hyp}{Hypothesis}{Hypotheses}

\makeatletter
\newcommand{\subalign}[1]{%
  \vcenter{%
    \Let@ \restore@math@cr \default@tag
    \baselineskip\fontdimen10 \scriptfont\tw@
    \advance\baselineskip\fontdimen12 \scriptfont\tw@
    \lineskip\thr@@\fontdimen8 \scriptfont\thr@@
    \lineskiplimit\lineskip
    \ialign{\hfil$\m@th\scriptstyle##$&$\m@th\scriptstyle{}##$\hfil\crcr
      #1\crcr
    }%
  }%
}
\makeatother

\newcommand{\F}{\mathbb{F}}                                         
\newcommand{\Fp}{\F_{p}}                                            
\newcommand{\Fq}{\F_{q}}                                            
\newcommand{\Fpn}{\Fp^{n}}                                          
\newcommand{\Fqn}{\Fq^{n}}                                          
\newcommand{\Fpnxn}{\Fp^{n \times n}}                               
\newcommand{\Fqm}{\Fq^{m}}                                          
\newcommand{\Fpx}{\Fp^\times}                                       

\DeclareMathOperator{\LP}{LP}                                       
\DeclareMathOperator{\Tr}{Tr}                                       

\DeclareMathOperator{\solvdeg}{sd}                                  

\DeclareMathOperator{\circulant}{circ}                              

\DeclareMathOperator{\CORR}{CORR}                                   
\DeclareMathOperator{\prob}{\mathbb{P}}                             

\newcommand*{\degree}[1]{\deg \left( #1 \right)}                    

\DeclareMathOperator{\wt}{wt}                                       

\newcommand{\Anemoi}{\texttt{Anemoi}\xspace}

\newcommand{\Arion}{\textsf{Arion}\xspace}
\newcommand{\ArionHash}{\textsf{ArionHash}\xspace}
\newcommand{\aArion}{\textalpha-\Arion}
\newcommand{\aArionHash}{\textalpha-\ArionHash}
\newcommand{\Arionpi}{\Arion-\textpi\xspace}

\newcommand{\GMiMC}{\texttt{GMiMC}\xspace}

\newcommand{\Griffin}{\textsc{Griffin}\xspace}

\newcommand{\Hades}{\textsc{Hades}\xspace}

\newcommand{\libsnark}{\texttt{libsnark}\xspace}

\newcommand{\LowMC}{\texttt{LowMC}\xspace}

\newcommand{\MiMC}{\texttt{MiMC}\xspace}

\newcommand{\Plonk}{\textsf{Plonk}\xspace}

\newcommand{\Poseidon}{\textsc{Poseidon}\xspace}

\newcommand{\RoneCS}{\textsf{R1CS}\xspace}

\newcommand{\ReinforcedConcrete}{\texttt{Reinforced Concrete}\xspace}

\newcolumntype{P}[1]{>{\centering\arraybackslash}p{#1}}
\newcolumntype{M}[1]{>{\centering\arraybackslash}m{#1}}

\DeclareMathOperator{\mb}{MB}

\begin{document}
    \title{\textsf{Arion}: Arithmetization-Oriented Permutation and Hashing from Generalized Triangular Dynamical Systems}
    \author{Arnab Roy \inst{1} \and
            Matthias Johann Steiner \inst{1} \and
            Stefano Trevisani \inst{2}}
    \institute{Alpen-Adria-Universit\"at Klagenfurt, Universit\"atsstraße 65-67, 9020 Klagenfurt am W\"orthersee, Austria \\ \email{firstname.lastname@aau.at} \and Technische Universit\"at Wien, Karlsplatz 13, 1040 Wien, Austria \\ \email{firstname.lastname@tuwien.ac.at}}

    \maketitle

    \begin{abstract}
        In this paper we propose the (keyed) permutation \textsf{Arion} and the hash function \textsf{ArionHash} over $\mathbb{F}_p$ for odd  and particularly large primes.
        The design of \textsf{Arion} is based on the newly introduced Generalized Triangular Dynamical System (GTDS), which provides a new algebraic framework for constructing (keyed) permutation using polynomials over a finite field.
        At round level \textsf{Arion} is the first design which is instantiated using the new GTDS.
        We provide extensive security analysis of our construction including algebraic cryptanalysis (e.g. interpolation and Gr\"obner basis attacks) that are particularly decisive in assessing the security of permutations and hash functions over $\mathbb{F}_p$.
        From an application perspective, \textsf{ArionHash} aims for efficient implementation in zkSNARK protocols and Zero-Knowledge proof systems.
        For this purpose, we exploit that CCZ-equivalence of graphs can lead to a more efficient implementation of Arithmetization-Oriented primitives.

        We compare the efficiency of \textsf{ArionHash} in \textsf{R1CS} and \textsf{Plonk} settings with other hash functions such as \textsc{Poseidon}, \texttt{Anemoi} and \textsc{Griffin}.
        For demonstrating the practical efficiency of \textsf{ArionHash} we implemented it with the zkSNARK libraries \texttt{libsnark} and Dusk Network \textsf{Plonk}.
        Our result shows that \textsf{ArionHash} is significantly faster than \textsc{Poseidon} - a hash function designed for zero-knowledge proof systems.
        We also found that an aggressive version of \textsf{ArionHash} is considerably faster than \texttt{Anemoi} and \textsc{Griffin} in a practical zkSNARK setting.
    \end{abstract}

    \section{Introduction}
    With the advancement of Zero-Knowledge (ZK), Multi-Party Computation (MPC) and Fully Homomorphic Encryption (FHE) in recent years new efficiency measures for symmetric-key primitives allowing efficient implementation in these schemes, namely low multiplicative complexity and low multiplicative depth, have been introduced.
    The block ciphers, permutations and hash functions with low multiplicative complexity are also referred to as Arithmetization-Oriented (AO) primitives.
    A significant number of these new types of AO primitives are defined over large finite fields of prime order $p \ggg 2$ for target applications.
    Our focus in this paper will be such a low multiplicative complexity construction over $\Fp$ for large primes.
    Some generic definitions and results in this paper are applicable to any odd prime, thus we describe these results and definitions accordingly.
    However, the security of the construction(s) will be analyzed only for large primes.

    To put this paper into context with previous AO constructions we give a short overview of their developments.
    The AO primitives proposed in the literature until now can be categorized into three generations.
    \begin{enumerate}[labelindent=*,leftmargin=*,label=Gen \Roman*:]
        \item \LowMC \cite{EC:ARSTZ15}, \MiMC \cite{AC:AGRRT16}

        \item \Hades \cite{EC:GLRRS20}, \Poseidon \cite{USENIX:GKRRS21}, \GMiMC \cite{ESORICS:AGPRRRRS19}, Rescue-Prime \cite{ToSC:AABDS20}

        \item \ReinforcedConcrete \cite{ReinforcedConcrete}, \Griffin \cite{Griffin}, \Anemoi \cite{Anemoi}, \Arion (this paper)
    \end{enumerate}

    The first generation consists of constructions which demonstrated that one can construct secure and efficient ciphers and hash functions with low degree primitives at round level.
    In particular, \LowMC introduced the partial Substitution Permutation Network (SPN) strategy in AO.

    In the second generation researchers tried to obtain further efficiency improvements from Feistel and (partial) SPNs to obtain new efficient primitives.
    Moreover, more focus was given on constructions native in large prime fields $\Fp$ rather than $\F_{2^n}$.
    This resulted in \Hades which combines full and partial SPNs over $\Fp$, and its derived sponge function \Poseidon which is now a widely deployed hash function for ZK applications.

    The current third generation adopts new design principles which neither reduce to the Feistel nor the SPN that culminated in the Generalized Triangular Dynamical System (GTDS) \cite{GTDS}.
    Moreover, this generation diverted from the consensus that one needs low degree polynomials to instantiate a secure and efficient AO primitive.

    In this paper we propose new AO primitives - \Arion (block cipher) and the hash function derived from it \ArionHash.
    At round level \Arion (and \ArionHash) like \Griffin, utilize(s) a polynomial of very high degree in one branch and low degree polynomials in the remaining branches to significantly cut the number of necessary rounds compared to the previous generations.
    \Anemoi also utilizes a high degree permutation, the so-called open \texttt{Flystel}, at round level, but to limit the number of constraints in a prover circuit the authors proved that the open \texttt{Flystel} is CCZ-equivalent (cf.\ \cite{Carlet-CCZ} and \cite[§4.2]{Anemoi}) to a low degree function, the so-called closed \texttt{Flystel}.
    Lastly, \ReinforcedConcrete is the first AO hash function that utilizes look-up tables which significantly reduces the number of necessary rounds of \ReinforcedConcrete and consequently also the number of constraints in a prover circuit.

    \subsection{Our Results}
    In this paper we propose the block cipher \Arion and the hash function \ArionHash (\Cref{Sec: permutation}), using the Generalized Triangular Dynamical System \cite{GTDS}.
    The block cipher and hash function are constructed over $\mathbb{F}_p$ with the target to achieve low multiplicative complexity in a prover circuit.
    Utilizing the structure of GTDS enables us to provide a systematic security analysis of the newly proposed block cipher and hash function. The GTDS structure also allows us to choose the best suited parameters for the efficiency.
    We provide extensive security analysis of the proposed block cipher and hash function against state-of-the-art cryptanalysis techniques to justify their security (\Cref{Sec: security analysis}).
    Our construction validates the soundness of the generic GTDS structure that uses polynomial dynamical system for constructing cryptographic permutations over finite fields.

    Although \Arion and \ArionHash are defined on arbitrary finite fields $\F_p$, the parameters of the block cipher and hash function are chosen in such way to be compatible with the primes chosen for the target ZK application namely, for $\text{BLS}12$ and $\text{BN}254$ curves.
    We propose aggressive versions of \Arion and \ArionHash namely, \aArion and \aArionHash.
    The difference between \Arion (and \ArionHash) and its aggressive version is that the former avoids a recently proposed probabilistic Gr\"obner basis \cite{Faugere-SubCubic} attack (\Cref{Sec: algebraic analysis} and Appendix C in the full version of the paper \cite{Arion}).

    To demonstrate and compare the efficiencies of our constructions (\Cref{Sec: performance}) we implemented them using the zkSNARK libraries \libsnark \cite{libsnark}, a \verb!C++! library used in the privacy protecting digital currency Zcash \cite{Zcash}, and Dusk Network \Plonk \cite{Dusk-Plonk}, a \verb|Rust| library used in the privacy-oriented blockchain protocol Dusk.
    Our results show that \ArionHash is significantly (2x) faster than \Poseidon\ - an efficient hash function designed for zkSNARK applications.
    The efficiency of \ArionHash is comparable to the recently proposed (but not yet published at a peer-reviewed venue) hash functions \Anemoi and \Griffin.
    We find that \aArionHash for practical choices of parameters in a Merkle tree mode of hashing is faster than \Griffin.
    We also reveal that CCZ-equivalence between the graphs of the \ArionHash GTDS and another closely related GTDS leads to a more efficient implementation of \ArionHash in ZK applications compared to the naive circuit for \ArionHash (\Cref{Sec: reducing the constraints}).

    Our public GitHub repository
    \begin{center}
        \url{https://github.com/sca-research/Arion}
    \end{center}
    contains reference implementations in \verb!SageMath!, \verb!C++! and \verb|Rust|, our \verb!OSCAR! implementation to perform Gr\"obner basis experiments, and our \verb!SageMath! code to estimate the security levels of \Arion \& \ArionHash.

    \section{The (Keyed) Permutation and Hash Function}\label{Sec: permutation}
    \subsection{Overview of the Design Rationale}
    Before we define \Arion and \ArionHash, we quickly summarize the design rationale behind our construction.
    \begin{itemize}
        \item By utilizing the GTDS to instantiate the permutation we aim to achieve fast degree growth in each component like in SPNs and non-linear mixing between the components as in Feistel network.
        Our GTDS, see \Cref{Def: GTDS}, incorporates the strength of both SPN and Feistel in a single primitive at round level.

        \item It follows from the generic security analysis in \cite[§5]{GTDS} that the univariate permutations, the SPN part, of the GTDS determine worst case security bounds against differential and linear cryptanalysis.
        Hence, we chose parameters that minimize these bounds.

        \item To thwart interpolation attacks we opted for a design that can achieve a degree overflow in the input variables in the first round, see \Cref{Lem: degrees in GTDS} and \Cref{Tab: degree overflow}.
        This is achieved in the SPN part of the GTDS by applying a low degree univariate permutation polynomial $p_1$ in all branches except the last one and by applying a high-degree inverse permutation $p_2^{-1}$ in the last branch.

        \item We opted for a linear layer that mixes all branches in every round. This is achieved with a circulant matrix which has only non-zero entries.

        \item For the high degree inverse permutation, the naive circuit for $p_2^{-1} (x) = y$ introduces many multiplicative constraints, though one can always transform such a circuit into a circuit for $x = p_2 (y)$ in constant time, see \Cref{Sec: reducing the constraints}.
        This trick plays a fundamental role in the efficiency of \ArionHash circuits.
    \end{itemize}

    \subsection{Keyed Permutation}
    We start with the definition of the generalized triangular dynamical system of \Arion.
    \begin{defn}[GTDS of \Arion]\label{Def: GTDS}
        Let $p \in \mathbb{Z}_{> 4}$ be a prime, and let $\Fp$ be the field with $p$ elements.
        Let $n, d_1, d_2, e \in \mathbb{Z}_{> 1}$ be integers such that
        \begin{enumerate}[label=(\roman*)]
            \item $d_1$ is the smallest positive integer such that $\gcd \left( d_1, p - 1 \right) = 1$,

            \item $d_2$ is an arbitrary integer such that $\gcd \left( d_2, p - 1 \right) = 1$, and

            \item $e \cdot d_2 \equiv 1 \mod p - 1$.
        \end{enumerate}
        For $1 \leq i \leq n - 1$ let $\alpha_{i, 1}, \alpha_{i, 2}, \beta_{i}\in \Fp$ be such that $\alpha_{i, 1}^2 - 4 \cdot \alpha_{i, 2}$ is a quadratic non-residue modulo $p$.
        The generalized triangular dynamical system $\mathcal{F}_\Arion = \{ f_1, \dots, f_n \}$ of {\Arion} is defined as
        \begin{equation*}
            \begin{split}
                f_i (x_1, \dots, x_n) &= x_i^{d_1} \cdot g_i (\sigma_{i + 1, n}) + h_i (\sigma_{i + 1, n}), \qquad 1 \leq i \leq n - 1, \\
                f_n (x_1, \dots, x_n) &= x_n^{e},
            \end{split}
        \end{equation*}
        where
        \begin{equation*}
            \begin{split}
                g_i (x) &= x^2 + \alpha_{i, 1} \cdot x + \alpha_{i, 2}, \\
                h_i (x) &= x^2 + \beta_i \cdot x, \\
                \sigma_{i + 1, n} &= \sum_{j = i + 1}^{n} x_j + f_j (x_1, \dots, x_n).
            \end{split}
        \end{equation*}
    \end{defn}

    Note that the GTDS $\mathcal{F} = \{ f_1, \dots, f_n \}$ must be considered as ordered tuple of polynomials since in general the order of the $f_i$'s cannot be interchanged.
    Since $\alpha_{i, 1}^2 - 4 \cdot \alpha_{i, 2}$ is a non-residue modulo $p$ for all $1 \leq i \leq n - 1$ the polynomials $g_i$ do not have any zeros over $\Fp$, therefore we can invert the GTDS with the procedure described in \cite[Proposition~8, Corollary~9]{GTDS}.
    In \Cref{Tab: efficient evaluation} we propose suitable exponents for $d_2$ which can be evaluated with at most $9$ multiplications.

    All exponents are chosen so that \Arion and \ArionHash provide at least 128 bit security against Gr\"obner basis attacks while minimizing the number of multiplicative constraints in a prover circuit, see \Cref{Sec: algebraic analysis,Sec: reducing the constraints}.
    \begin{table}[hb]
        \centering
        \caption{Efficient evaluation of exponents $d_2 \in \{ 121, 123, 125, 129, 161, 193, 195, 257 \}$.}
        \label{Tab: efficient evaluation}
        \resizebox{0.8\textwidth}{!}{
            \begin{tabular}{ M{7mm} | M{7.5cm} | M{2.5cm} }
                \toprule
                $d_2$ & Evaluation chain & Number of Multiplications \\
                \midrule

                $121$ & $y = \left( x^2 \right)^2$, $z = \left( y^2 \cdot y \right)^2$, $x^{121} = \left( z^2 \right)^2 \cdot z \cdot x$                                  & $9$ \\[7pt]
                $123$ & $y = x^2 \cdot x$, $z = \left( \left( y^2 \right)^2 \right)^2$, $x^{123} = \left( z^2 \right)^2 \cdot z \cdot y$                                  & $9$ \\[7pt]
                $125$ & $y = \left( x^2 \right)^2 \cdot x$, $z = \left( \left( y^2 \right)^2 \right)^2$, $x^{125} = z^2 \cdot z \cdot y$                                  & $9$ \\[7pt]
                $129$ & $y = \left( \left( \left( x^2 \right)^2 \right)^2 \right)^2$, $z = \left( \left( y^2 \right)^2 \right)^2$, $x^{129} = z \cdot x$                  & $8$ \\[7pt]
                $161$ & $y = \left( x^2 \right)^2 \cdot x$, $z = \left(\left( y^2 \right)^2 \right)^2$, $x^{161} = \left( z^2 \right)^2 \cdot x$                          & $9$ \\[7pt]
                $193$ & $y = \left( x^2 \right) \cdot x$, $z = \left( \left( \left( y^2 \right)^2 \right)^2 \right)^2$, $x^{193} = \left( z^2 \right)^2 \cdot x$          & $9$ \\[7pt]
                $195$ & $y = \left( x^2 \right) \cdot x$, $z = \left( \left( \left( y^2 \right)^2 \right)^2 \right)^2$, $x^{195} = \left( z^2 \right)^2 \cdot y$          & $9$ \\[7pt]
                $257$ & $y = \left( \left( \left( x^2 \right)^2 \right)^2 \right)^2$, $z = \left( \left( \left( y^2 \right)^2 \right)^2 \right)^2$, $x^{257} = z \cdot x$ & $9$ \\

                \bottomrule
            \end{tabular}
        }
    \end{table}

    Let us compute the degrees of the polynomials in the GTDS.
    \begin{lem}\label{Lem: degrees in GTDS}
        Let $n, d_1, e \geq 1$ be integers, and let $\mathcal{F}_\Arion = \{ f_1, \dots, f_n \}$ be an \Arion GTDS.
        Then
        \begin{equation*}
            \degree{f_i} = 2^{n - i} \cdot (d_1 + e) - d_1.
        \end{equation*}
    \end{lem}
    \begin{proof}
        We perform an upwards induction, for $n$ and $n - 1$ the claim is clear.
        Suppose the claim is true for indices greater than or equal to $i$, i.e.\ $\degree{f_i} = \degree{f_i} = 2^{n - i} \cdot (d_1 + e) - d_1$.
        By construction, the leading monomial of $f_{i - 1}$ is the leading term of the polynomial $x_{i - 1}^{d_1} \cdot f_i^2$.
        Thus,
        \begin{align*}
            \degree{f_{i - 1}}
            &= \degree{x_{i - 1}^{d_1} \cdot f_i^2}
            = d_1 + 2 \cdot \left( 2^{n - i} \cdot (d_1 + e) - d_1 \right) \\
            &= 2^{n - (i - 1)} \cdot (d_1 + e) - d_1,
        \end{align*}
        which proves the claim. \qed
    \end{proof}

    To introduce mixing between the blocks we chose a circulant matrix whose product with a vector can be efficiently evaluated.
    \begin{defn}[Affine layer of \Arion]\label{Def: affine layer}
        Let $p \in \mathbb{Z}$ be a prime and let $\Fp$ be the field with $p$ elements.
        The affine layer of {\Arion} is defined as
        \begin{equation*}
            \mathcal{L}_\mathbf{c}: \Fpn \to \Fpn, \quad \mathbf{x} \mapsto \circulant \left(1, \dots, n\right) \mathbf{x} + \mathbf{c},
        \end{equation*}
        where $\circulant \left( 1, \dots, n \right) \in \Fpnxn$ is the circulant matrix\footnote{
            We shortly recall the definition of (right) circulant matrices: Let $k$ be a field and let $\mathbf{v} = (v_1, \dots, v_n) \in k^n$, then the circulant matrix of $\mathbf{v}$ is defined as
            \begin{equation*}
                \circulant (\mathbf{v}) =
                \begin{pmatrix}
                    v_1 & v_2 & \ldots & v_{n - 1} & v_n        \\
                    v_n & v_1 & \ldots & v_{n - 2} & v_{n - 1}  \\
                    &     & \ddots &           &            \\
                    v_2 & v_3 & \ldots & v_n       & v_1
                \end{pmatrix}
                .
            \end{equation*}
        } with entries $1, \dots, n$ in the first row and $\mathbf{c} \in \Fpn$ is a constant vector.
    \end{defn}
    \begin{rem}
        For any prime number $p \in \mathbb{Z}$ with $p > 130$ and $n = 2, 3, 4$ the matrix $\circulant \left( 1, \dots, n \right)$ is a MDS matrix over $\Fp$.
    \end{rem}

    The following algorithm provides an efficient way to evaluate the matrix-vector product for $\circulant \left( 1, \dots, n \right)$.
    \begin{algorithm}[H]
        \caption{Efficient evaluation of matrix-vector product}
        \label{Alg: efficient algorithm}
        \begin{algorithmic}[1]
            \Statex \textbf{Input}
            \StateX $\mathbf{v} = \left( v_1, \dots, v_n \right)^\intercal \in \Fpn$

            \Statex \textbf{Output}
            \StateX $\circulant \left( 1, \dots, n \right) \mathbf{v} \in \Fpn$

            \State Initialize $\mathbf{w} = \left( 0, \dots, 0 \right) \in \Fpn$.

            \State Compute $\sigma = \sum_{i = 1}^{n} v_i$ and set $w_1 = \sigma + \sum_{i = 1}^{n} (i - 1) \cdot v_i$.

            \State Set $i = 2$.

            \While{$i \leq n$}

            \State Set $w_i = w_{i - 1} - \sigma + n \cdot v_{i - 1}$.

            \State $i = i + 1$.

            \EndWhile

            \State \Return $\mathbf{w}$
        \end{algorithmic}
    \end{algorithm}

    To define a keyed permutation we need a key addition which we denote as
    \begin{equation*}
        \mathcal{K}_\mathbf{k} : \Fpn \times \Fpn \to \Fpn, \qquad (\mathbf{x}, \mathbf{k}) \mapsto \mathbf{x} + \mathbf{k}.
    \end{equation*}
    The keyed permutation \Arion is now defined as follows.
    \begin{defn}[\Arion]
        Let $p \in \mathbb{Z}$ be a prime and let $\Fp$ be the field with $p$ elements, and let $n > 1$ and $r \geq 1$ be integers.
        For $1 \leq i \leq r$ let $\mathcal{F}^{(i)}_\Arion : \Fpn \to \Fpn$ be an \Arion GTDS and for $1 \leq i \leq r$ let $\mathcal{L}_{\mathbf{c}_i}: \Fpn \to \Fpn$ be affine layers from \Cref{Def: affine layer}.
        The $i$\textsuperscript{th} round function of {\Arion} is defined as
        \begin{equation*}
            \begin{split}
                \mathcal{R}_{\mathbf{k}}^{(i)}: \Fpn &\times \Fpn \to \Fpn, \\
                \left( \mathbf{x}, \mathbf{k} \right) &\mapsto \mathcal{K}_\mathbf{k} \circ \mathcal{L}_{\mathbf{c}_i} \circ \mathcal{F}^{(i)}_\Arion \left( \mathbf{x} \right).
            \end{split}
        \end{equation*}
        Then \Arion is defined as the following composition
        \begin{equation*}
            \begin{split}
                \Arion: \Fpn \times \Fp^{n \times (r + 1)} &\to \Fpn, \\
                \big( \mathbf{x}, \mathbf{k}_0, \mathbf{k}_1, \dots, \mathbf{k}_r \big) &\mapsto \mathcal{R}_{\mathbf{k}_r}^{(r)} \circ \cdots \circ \mathcal{R}_{\mathbf{k}_1}^{(1)} \circ \mathcal{L}_{\mathbf{0}} \circ \mathcal{K}_{\mathbf{k}_0} \left( \mathbf{x} \right).
            \end{split}
        \end{equation*}
    \end{defn}
    Further, we denote with \Arionpi the unkeyed permutation, i.e.\ \Arion is instantiated with the key $\mathbf{k}_0 = \ldots = \mathbf{k}_r = \mathbf{0}$.

    Since our final aim is to construct a hash function using the \Arionpi, we analyze \Arion only for keys $\mathbf{k}_j = \mathbf{k}$, where $\mathbf{k} \in \Fpn$, in every round.
    We do not give any key scheduling algorithm for keys whose sizes are larger than the block size. Instantiation of \Arion with such a key is not the main topic of this paper.

    \subsection{Hash Function}
    For the hash function \ArionHash over $\Fpn$ we instantiate \Arionpi in sponge mode \cite{Sponge,EC:BDPV08}.
    The state size $n = r + c$ is split into the rate part $r$ and the capacity part $c$.
    In \cite[Theorem~2]{EC:BDPV08} it has been proven that for a random permutation the sponge construction is indistinguishable from a random oracle up to $\min \left\{ p^r, p^{c / 2} \right\}$ queries.
    Therefore, to provide $\kappa$ bits of security, $p^r \geq 2^\kappa$ and $p^{c / 2} \geq 2^\kappa$, we require that
    \begin{equation}
        r \geq \frac{\kappa}{\log_2 \left( p \right)}, \qquad\text{and}\qquad c \geq \frac{2 \cdot \kappa}{\log_2 \left( p \right)}.
    \end{equation}
    Given an input message $m$ we choose a similar padding rule as for \Poseidon \cite[§4.2]{USENIX:GKRRS21}, we add the smallest number of zeros $< r$ such that the size of $m \mid\mid 0^\ast$ is a multiple of $r$.
    If we have to pad the message, then we replace the initial value $\texttt{IV} \in \Fp^c$ with $|m| \mid\mid \texttt{IV}' \in \Fp^c$, where $|m| \in \Fp$ is the size of the input message $m$ and $\texttt{IV}' \in \Fp^{c - 1}$ is an initial value.

    \subsection{Instantiations}
    Target primes for \Arion \& \ArionHash are the $250$ bit prime numbers $\text{BLS}12$ and $\text{BN}254$\footnote{
        ${\scriptstyle\text{BLS}12 = 0x73eda753299d7d483339d80809a1d80553bda402fffe5bfeffffffff00000001}$, \\
        ${\scriptstyle\text{BN}254 = 0x30644e72e131a029b85045b68181585d2833e84879b9709143e1f593f0000001}$.}.
    Since degree growth of the \Arion GTDS is dominated by the power permutation in the bottom component we list the smallest integer $m \in \mathbb{Z}$ such that $m \cdot e \geq p$ in \Cref{Tab: degree overflow}.
    \begin{table}[H]
        \centering
        \caption{Smallest positive integer so that $m \in \mathbb{Z}$ such that $m \cdot e \geq p$ for $\text{BLS}12$ and $\text{BN}254$.}
        \label{Tab: degree overflow}
        \begin{tabular}{ M{7mm} | M{1.2cm} | M{1.2cm} }
            \toprule
            & \multicolumn{2}{ c }{$\left\lceil p / e \right\rceil$} \\
            \midrule
            $d_2$ & $\text{BLS}12$ & $\text{BN}254$ \\
            \midrule
            $121$ & n.a.\ & $3$   \\
            $123$ & n.a.\ & n.a.\ \\
            $125$ & $3$   & $2$   \\
            $129$ & n.a.\ & n.a.\ \\
            $161$ & $3$   & $4$   \\
            $193$ & $13$  & $3$   \\
            $195$ & n.a.\ & n.a.\ \\
            $257$ & $3$   & $2$   \\
            \bottomrule
        \end{tabular}
    \end{table}
    Therefore, by \Cref{Lem: degrees in GTDS} for $n = 3$ and all exponents except $d_2 = 193$ the \Arion GTDS surpasses degree $p$ in the first component.

    In \Cref{Tab: Arion parameters} we provide the parameters for \Arion and \ArionHash and their aggressive versions \aArion and \aArionHash with $d_1, d_2 \in \mathbb{Z}$ such that $\gcd \left( d_i, p - 1 \right) = 1$, $d_1 = 3, 5$ and $121 \leq d_2 \leq 257$.
    The number of rounds for \Arion and \ArionHash are chosen to provide 128 bit security against the most efficient probabilistic algorithm (available till date) for polynomial system solving in a Gr\"obner basis attack on \ArionHash.
    Since $2^{\frac{250}{2}} = 2^{125}$ we consider all possible rate-capacity pairs $n = c + r$ suitable for \ArionHash over $\text{BLS}12$ and $\text{BN}254$.
    As hash output of \ArionHash over $\text{BLS}12$ and $\text{BN}254$ we recommend to use a single $\Fp$ element.
    \begin{table}[H]
        \centering
        \caption{\Arion, \ArionHash, \aArion and \aArionHash parameters over primes $p \geq 2^{250}$ with $d_1, d_2 \in \mathbb{Z}$ such that $\gcd \left( d_i, p - 1 \right) = 1$ and $121 \leq d_2 \leq 257$ and 128 bit security.}
        \label{Tab: Arion parameters}
        \begin{tabular}[t]{ M{1.2cm} | M{3cm} | M{3.5cm}  }
            \toprule
            & \multicolumn{2}{ c }{Rounds} \\
            \midrule

            Blocks & \Arion \& \ArionHash & \aArion \& \aArionHash \\
            \midrule

            \multicolumn{3}{ c }{$d_1 = 3$} \\
            \midrule

            $3$ & $6$ & $5$ \\
            $4$ & $6$ & $4$ \\
            $5$ & $5$ & $4$ \\
            $6$ & $5$ & $4$ \\
            $8$ & $4$ & $4$ \\

            \midrule
            \multicolumn{3}{ c }{$d_1 = 5$} \\
            \midrule

            $3$ & $6$ & $4$ \\
            $4$ & $5$ & $4$ \\
            $5$ & $5$ & $4$ \\
            $6$ & $5$ & $4$ \\
            $8$ & $4$ & $4$ \\

            \bottomrule
        \end{tabular}
    \end{table}
    The number of rounds for \aArion and \aArionHash are chosen to provide 128 bit security against the most efficient deterministic algorithm for polynomial system solving in a Gr\"obner basis attack on \ArionHash.
    For more details on the security with respect to Gr\"obner basis attacks we refer to Appendix C in the full version of the paper \cite{Arion}.

    \section{Security Analysis of \Arion}\label{Sec: security analysis}
    \subsection{Statistical Cryptanalysis}
    \subsubsection{Differential Cryptanalysis.}
    In differential cryptanalysis \cite{JC:BihSha91} and its variants the propagation of input differences through the rounds of a block cipher or hash function is exploited to recover the secret key or to construct a collision.
    For the \Arion GTDS the probability that an input difference $\boldsymbol{\Delta x} \in \Fqn \setminus \{ \mathbf{0} \}$ propagates to the output difference is $\boldsymbol{\Delta y} \in \Fqn$ is bounded by (see \cite[Theorem~18, Corollary~19]{GTDS})
    \begin{equation}\label{Equ: differntial probability bound}
        \prob \left[ \mathcal{F}_\Arion \!: \boldsymbol{\Delta x} \to \boldsymbol{\Delta y} \right] \leq \left( \frac{d_2}{p} \right)^{\wt \left( \boldsymbol{\Delta x} \right)} \leq \frac{d_2}{p},
    \end{equation}
    where $\wt: \Fqn \to \mathbb{Z}$ denotes Hamming weight.
    For $p \geq 2^{250}$ and $d_2 \leq 2^9$ this probability is bounded by $2^{-241 \cdot \wt \left( \boldsymbol{\Delta x} \right)} \leq 2^{-241}$.
    Under the assumption that the rounds of \Arion are statistically independent we can estimate the probability of any non-trivial differential trail via $2^{-241 \cdot r}$.
    Moreover, even if an adversary can search a restricted differential hull of size $2^{120}$ between the $2$\textsuperscript{nd} and the $r$\textsuperscript{th} round, then two rounds are already sufficient to provide $128$ bit security against differential cryptanalysis.
    For more details we refer to Appendix A.1 in the full version of the paper \cite{Arion}.

    Note that a small differential probability also rules out the boomerang attack \cite{FSE:Wagner99,C:JouPey07} which exploits two complementary differential patterns that span the cipher, of which one must at least cover two rounds.

    \subsubsection{Truncated Differential \& Rebound Cryptanalysis.}
    In a truncated differential attack \cite{FSE:Knudsen94} an attacker can only predict parts of the difference between pairs of text.
    We expect that the \Arion GTDS admits truncated differentials of Hamming weight $1$ with probability $1$ for the first round.
    On the other hand, if $\wt (\mathbf{v}) = 1$, then we have that $\wt \big( \circulant (1, \dots, n) \mathbf{v} \big) = n$.
    Therefore, such a truncated differential activates all inputs in the second round of \Arion.
    Hence, for $p \geq 2^{250}$ and $d_2 \leq 2^{9}$ the differential probability for the second round is bounded by $2^{-250 \cdot n}$.
    Even if an adversary can search restricted differential hulls of size $2^{120}$ after the first round, this probability and \Cref{Equ: differntial probability bound} nullify truncated differential attacks within the $128$ bit security target.

    In a rebound attack \cite{AC:LMRRS09,FSE:MRST09} an adversary connects two (truncated) differentials in the middle of a cipher or hash function.
    Probability $1$ truncated differentials can cover at most one round of \Arion, so $r - 2$ rounds can be covered with an inside-out approach.
    By our previous analysis we do not expect that a successful rebound attack can be mounted on $4$ or more rounds on \Arion \& \ArionHash within the $128$ bit security target.

    For more details we refer to Appendix A.2 in the full version of the paper \cite{Arion}.

    \subsubsection{Linear Cryptanalysis.}
    Linear cryptanalysis \cite{EC:Matsui93,SAC:BaiSteVau07} utilizes affine approximations of round functions for a sample of known plaintexts.
    For any additive character $\chi: \Fq \to \mathbb{C}$ and any affine approximation $\mathbf{a}, \mathbf{b} \in \Fqn \setminus \{ \mathbf{0} \}$, the linear probability of the \Arion GTDS is bounded by (see \cite[Theorem~24, Corollary~25]{GTDS})
    \begin{equation}
        \LP_{\mathcal{F}_\Arion} (\chi, \mathbf{a}, \mathbf{b}) \leq \frac{\left( d_2 - 1 \right)^2}{q}.
    \end{equation}
    Therefore, for $p \geq 2^{250}$ and $d_2 \leq 2^9$ this probability is bounded by $2^{-232}$, and henceforth under the assumption of statistically independent rounds of \Arion the linear probability of any non-trivial linear trail is bounded by $2^{-232 \cdot r}$.
    Moreover, even if an adversary can search a restricted linear hull of size $2^{120}$ between the $2$\textsuperscript{nd} and the $r$\textsuperscript{th} round, then two rounds are already sufficient to provide $128$ bit security against linear cryptanalysis.
    For more details we refer to Appendix A.3 in the full version of the paper \cite{Arion}.

    \subsection{Algebraic Cryptanalysis}\label{Sec: algebraic analysis}
    \subsubsection{Interpolation \& Integral Cryptanalysis.}
    Interpolation attacks \cite{FSE:JakKnu97} construct the polynomial vector representing a cipher without knowledge of the secret key.
    If such an attack is successful against a cipher, then an adversary can encrypt any plaintext without knowledge of the secret key.
    Recall that any function $F: \Fqn \to \Fq$ can be represented by a polynomial $f \in \Fp [\mathbf{X}_n] = \Fq [x_1, \dots, x_n] / \left( x_1^q - x_1, \dots, x_n^q - x_n \right)$, thus at most $q^n$ monomials can be present in $f$.
    After the first round of \Arionpi we expect that the terms
    \begin{equation}
        \left( \sum_{i = 1}^{n} i \cdot x_i \right)^e + \sum_{i = 1}^{n} i \cdot x_i
    \end{equation}
    are present in every branch.
    After another application of the round function we expect to produce the terms
    \begin{equation}
        \left( \left( \sum_{i = 1}^{n} i \cdot x_i \right)^e + \sum_{i = 1}^{n} i \cdot x_i \right)^e \mod \left( x_1^p - x_1, \dots, x_n^p - x_n \right)
    \end{equation}
    in every branch.
    By our specification $e$ is the inverse exponent of a relatively low degree permutation, therefore we expect that after two rounds almost all monomials from $\Fp [\mathbf{X}_n]$ are present in every component of \Arion.
    For more details we refer to Appendix B.1 in the full version of the paper \cite{Arion}.

    For a polynomial $f \in \Fq [x_1, \dots, x_n]$ an integral distinguisher \cite{FSE:KnuWag02,C:BCDELLNPSTW20} exploits that for any affine subspace $V \subset \Fqn$ with $\degree{f} < \dim \left( V \right) \cdot \left( q - 1 \right)$ one has that
    \begin{equation}
        \sum_{\mathbf{x} \in V} f (\mathbf{x}) = 0.
    \end{equation}
    If almost all monomials are present in \Arionpi, then $\degree{\text{\Arionpi}} \approx n \cdot (q - 1)$ in every component, so only $V = \Fqn$ is a suitable subspace for an integral distinguisher.
    Therefore, we do not expect that non-trivial integral distinguishers exist for \Arionpi.
    For more details we refer to Appendix B.2 in the full version of the paper \cite{Arion}.

    \subsubsection{Gr\"obner Basis Analysis.}
    In a Gr\"obner basis attack \cite{Buchberger,Cox-Ideals} one models a cipher or hash function as fully determined polynomial system and then solves for the key or preimage.
    For Gr\"obner basis analysis of \Arion \& \ArionHash we assume that a degree reverse lexicographic (DRL) Gr\"obner basis can be found in $\mathcal{O} (1)$.
    We base the security of \Arion \& \ArionHash solely on the complexity of solving their polynomial systems via state of the art deterministic and probabilistic Gr\"obner basis conversion algorithms \cite{Faugere-FGLM,Faugere-SparseFGLM,Faugere-SubCubic} combined with the univariate polynomial solving algorithm of Bariant et al.\ \cite[§3.1]{ToSC:BBLP22}.
    With this methods, solving a fully determined polynomial system over a finite field $\Fq$  with known DRL Gr\"obner basis via deterministic methods requires
    \begin{equation}
        \mathcal{O} \left( n \cdot d^\omega + d \cdot \log \left( q \right) \cdot \log \left( d \right) \cdot \log \big( \log \left( d \right) \big) + d \cdot \log \left( d \right)^2 \cdot \log \big( \log \left( d \right) \big) \right)
    \end{equation}
    field operations, and with deterministic methods
    \begin{equation}
        \mathcal{O} \left( \sqrt{n} \cdot d^{2 + \frac{n - 1}{n}} + d \cdot \log \left( q \right) \cdot \log \left( d \right) \cdot \log \big( \log \left( d \right) \big) + d \cdot \log \left( d \right)^2 \cdot \log \big( \log \left( d \right) \big) \right)
    \end{equation}
    field operations, where $q$ is the size of the finite field, $n$ is the number of variables, $d$ is the $\Fq$-vector space dimension of the polynomial ring modulo the polynomial system and $2 \leq \omega < 2.3727$ is a linear algebra constant.
    We conjecture that the quotient space dimension of \Arion grows or bounded by
    \begin{equation}
        \dim_{\Fp} \left( \mathcal{F}_\Arion \right) (n, r, d_1, d_2) = \left( d_2 \cdot \left( d_1 + 2 \right)^{n - 1} \right)^r,
    \end{equation}
    and for \ArionHash we conjecture that the dimension grows or is bounded by
    \begin{equation}
        \dim_{\Fp} \left( \mathcal{F}_\ArionHash \right) (n, r, d_1, d_2) = \Big( 2^{n - 1} \cdot d_2 \cdot \left( d_1 + 1 \right) - d_1 \cdot d_2 \Big)^r.
    \end{equation}
    Round numbers for \Arion \& \ArionHash in \Cref{Tab: Arion parameters} are chosen to resist deterministic as well as probabilistic Gr\"obner basis attacks against an ideal adversary with $\omega = 2$ within the 128 bit security target.
    Round numbers for \aArion \& \aArionHash in \Cref{Tab: Arion parameters} are chosen to resist only deterministic Gr\"obner basis attacks within the 128 bit security target.

    For \ArionHash one can set up a collision polynomial system by connecting two preimage polynomial systems.
    Note that this polynomial system is in general not fully determined, therefore an adversary has to randomly guess some variables before solving the system.
    If an adversary guesses output variables of the sponge until the collision polynomial system is fully determined, then we conjecture that the quotient space dimension of the collision polynomial system grows or is bounded by
    \begin{equation}
        \dim_{\Fp} \left( \mathcal{F}_{\ArionHash, coll} \right) \left( n, r, d_1, d_2 \right) = \left( \dim_{\Fp} \left( \mathcal{F}_\ArionHash \right) \left( n, r, d_1, d_2 \right) \right)^2.
    \end{equation}
    Thus, we do not expect a collision Gr\"obner basis attack to be more performative than a preimage attack.
    For more details we refer to Appendix C in the full version of the paper \cite{Arion}.

    \section{Performance Evaluation}\label{Sec: performance}
    In this section, we compare various instances of \ArionHash, \Anemoi, \Griffin and \Poseidon with respect to \RoneCS (\Cref{Sec: R1CS}) and \Plonk (\Cref{Sec: Plonk}).
    For starters, we discuss the theoretical foundation of an efficient implementation of an \ArionHash circuit.
    In the \Anemoi proposal it was revealed that CCZ-equivalence is a route to construct high degree permutations that can be verified with CCZ-equivalent low degree functions \cite[§4.1]{Anemoi}.
    In \Cref{Sec: reducing the constraints} we follow this approach to prove that a \ArionHash circuit can be transformed into an efficient circuit that avoids the computation of $x^e$ via an affine transformation.

    \subsection{Reducing the Number of Constraints}\label{Sec: reducing the constraints}
    By definition of the \Arion GTDS a prover circuit will have to verify that
    \begin{equation}\label{Equ: naive circuit}
        y = x^e,
    \end{equation}
    though since $e$ induces the inverse power permutation to $d_2 \in \{ 121, 123, 125, 129, \allowbreak 161, 193, 195, 257 \}$ the naive circuit for \Cref{Equ: naive circuit} will introduce many constraints.
    On the other hand, from a prover's perspective \Cref{Equ: naive circuit} is equivalent to
    \begin{equation}\label{Equ: equivalent circuit}
        y^{d_2} = \left( x^e \right)^{d_2} = x,
    \end{equation}
    for all $x \in \Fp$.
    Thus, in an implementation to reduce the number of multiplicative constraints we are well advised to implement the equivalent circuit instead of the naive circuit.
    We also would like to note that the same trick was applied in \texttt{Griffin} \cite{Griffin} to reduce the number of constraints.

    In the design of \Anemoi \cite[§4]{Anemoi} a new tool was introduced to reduce the number constraints for an \Anemoi circuit: CCZ-equivalence \cite{Carlet-CCZ}.
    The authors have found a high degree permutation, the open \texttt{Flystel}, which is CCZ-equivalent to a low degree function, the closed \texttt{Flystel}.
    Consequently, this can be exploited to significantly reduce the number of constraints in a prover circuit (cf.\ \cite[Corollary~2]{Anemoi}).
    Let us now formalize the trick in \Cref{Equ: equivalent circuit} in terms of CCZ-equivalence.
    \begin{defn}
        Let $\Fq$ be a finite field, and let $F, G: \Fqn \to \Fqm$ be functions.
        \begin{enumerate}
            \item The graph of $F$ is defined as
            \begin{equation*}
                \Gamma_F = \Big\{ \big(\mathbf{x}, F (\mathbf{x}) \big) \mid \mathbf{x} \in \Fqn \Big\}.
            \end{equation*}

            \item $F$ and $G$ are said to be CCZ-equivalent if there exists an affine permutation $\mathcal{A}$ of $\Fqn \times \Fqm$ such that
            \begin{equation*}
                \Gamma_F = \mathcal{A} (\Gamma_G).
            \end{equation*}
        \end{enumerate}
    \end{defn}

    Now let us describe a GTDS that is equivalent to the \Arion GTDS.
    \begin{prop}\label{Prop: CCZ-equivalence}
        Let $\Fp$ be a prime field, and let $n, d_1, d_2, e \in \mathbb{Z}_{> 1}$ be integers such that
        \begin{enumerate}[label=(\roman*)]
            \item $d_1$ is the smallest positive integer such that $\gcd \left( d_1, p - 1 \right) = 1$,

            \item $d_2$ is an arbitrary integer such that $\gcd \left( d_2, p - 1 \right) = 1$, and

            \item $e \cdot d_2 \equiv 1 \mod p - 1$.
        \end{enumerate}
        Let $\mathcal{F}_\Arion = \{ f_1, \dots, f_n \}$ be the \Arion GTDS, let $g_i, h_i \in \Fq [x]$ be the
        polynomials that define $\mathcal{F}_\Arion$, and let the GTDS
        $\mathcal{G} = \{ \hat{f}_1, \dots, \hat{f}_n \}$ be defined as
        \begin{equation*}
            \begin{split}
                \hat{f}_i (x_1, \dots, x_n) &= x_i^{d_1} \cdot g_i (\tau_{i + 1, n}) +
                h_i (\tau_{i + 1, n}), \qquad 1 \leq i \leq n - 1, \\
                \hat{f}_n (x_1, \dots, x_n) &= x_n^{d_2},
            \end{split}
        \end{equation*}
        where
        \begin{equation*}
            \tau_{i + 1, n} = \sum_{j = i + 1}^{n} x_j + \hat{f}_j (x_1, \dots, x_n).
        \end{equation*}
        Then $\mathcal{F}_\Arion$ is CCZ-equivalent to $\mathcal{G}$.
    \end{prop}
    \begin{proof}
        We consider the affine permutation $\mathcal{A}: \Fq^{2n} \to \Fq^{2n}$ that swaps the $n$\textsuperscript{th} element with the $(2n)$\textsuperscript{th} element, moreover we consider the substitution $\mathbf{x} = \big( \hat{x}_1, \dots, \hat{x}_{n - 1}, \hat{x}_n^{d_2} \big)$.
        Now we apply the affine permutation to $\Gamma_{\mathcal{F}_\Arion}$ which yields
        \begin{equation*}
            \mathcal{A} \big( \mathbf{x}, \mathcal{F}_\Arion (\mathbf{x}) \big) =
            \begin{pmatrix}
                \{ \hat{x}_i \}_{1 \leq i \leq n - 1} \\[2pt]
                \hat{x}_n^{e \cdot d_2} \\[2pt]
                \Big\{ f_i \left( \hat{x}_i, \dots, \hat{x}_n^{d_2} \right) \Big\}_{1 \leq i \leq n - 1} \\[2pt]
                \hat{x}_n^{d_2}
            \end{pmatrix}
            .
        \end{equation*}
        By construction of $d_2$ and $e$ we have that $x^{e \cdot d_2} = x$ for every $x \in \Fp$.
        Let's now investigate what happens to the $f_i$'s.
        Starting with $f_{n - 1}$, we have that
        \begin{equation*}
            \sigma_{n, n} \left( \hat{x}_n^{d_2} \right) = \hat{x}_n^{d_2} + \hat{x}_n^{e \cdot d_2} = \hat{x}_n + \hat{x}_n^{d_2} = \tau_{n, n} (\hat{x}_n),
        \end{equation*}
        for all $\hat{x}_n \in \Fp$, and therefore
        \begin{equation*}
            f_{n - 1} \left( \hat{x}_{n - 1}, \hat{x}_n^{d_2} \right) = \hat{f}_{n - 1} (\hat{x}_{n - 1}, \hat{x}_n).
        \end{equation*}
        Inductively, we now go through all the branches to conclude that $f_i \left( \hat{x}_i, \dots, \hat{x}_n^{d_2} \right) = \hat{f}_i (\hat{x}_i, \dots, \hat{x}_n)$ for all $1 \leq i \leq n - 1$ which proves that $\mathcal{A} \big( \mathbf{x}, \mathcal{F} (\mathbf{x}) \big) = \big( \boldsymbol{\hat{x}}, \mathcal{G} (\boldsymbol{\hat{x}}) \big)$. \qed
    \end{proof}

    \begin{cor}\label{Cor: prover transformation}
        Verifying that $(y_1, \dots, y_n) = \mathcal{F} (x_1, \dots, x_n)$ is equivalent to verifying that
        $(y_1, \dots, y_{n - 1}, x_n) = \mathcal{G} (x_1, \dots, x_{n - 1}, y_n)$.
    \end{cor}

    Note that it follows from \cite[Theorem~18, 24]{GTDS} that the \Arion GTDS $\mathcal{F}$ and its CCZ-equivalent GTDS $\mathcal{G}$ from \Cref{Prop: CCZ-equivalence} are in the same security class with respect to differential and linear cryptanalysis.
    Unlike as for \Anemoi the CCZ-equivalent GTDS $\mathcal{G}$ is not a low degree function,
    though when implementing it as prover circuit we never have use multiplications to compute $\tau_{i + 1, n}$.

    \subsection{\RoneCS Performance of \ArionHash}\label{Sec: R1CS}
    Estimating the number of multiplicative constraints in a \RoneCS circuit for \ArionHash is straightforward.
    \begin{lem}
        Let $\Fp$ be a finite field, let $r, n \geq 2$ be integers, and let $\ArionHash$ with $r$ rounds be
        defined over $\Fpn$.
        For $i = 1, 2$ denote with $d_{i, \text{inc}}$ the minimal number of multiplications to compute the univariate power
        permutation $x^{d_i}$.
        Then a prover \RoneCS circuit for \ArionHash needs
        \begin{equation*}
            N_\ArionHash = r \cdot \big( \left( n - 1 \right) \cdot \left( d_{1, \text{inc}} + 2 \right) + d_{2, \text{inc}} \big)
        \end{equation*}
        multiplicative constraints.
    \end{lem}
    \begin{proof}
        By \Cref{Cor: prover transformation} one needs $d_{2, \text{inc}}$ constraints in the
        $n$\textsuperscript{th} branch.
        In each of the remaining $n - 1$ branches one needs $d_{1, \text{inc}}$ constraints for the power
        permutation, $1$ constraints for the computation of $g_i$ and $h_i$ and $1$ multiplication for
        the product of the power permutation and $g_i$. \qed
    \end{proof}

    Analog the number of \RoneCS constraints for \Anemoi, \Griffin and \Poseidon (cf.\ \cite[\S 7]{Anemoi}, \cite[\S 7.2]{Griffin} and \cite{USENIX:GKRRS21}) are given by
    \begin{align}
        N_\Griffin &= 2 \cdot r \cdot \left( d_{inc} + n - 2 \right), \\
        N_\Anemoi &= \frac{r \cdot n}{2} \cdot \left( d_{inc} + 2\right), \\
        N_\Poseidon &= d_{inc} \cdot \left( n \cdot r_f + r_p \right).
    \end{align}
    In \Cref{Tab: round numbers} we compiled the round numbers of the hash functions.
    \clearpage 
    \begin{table}[H]
        \centering
        \caption{Round numbers for \Anemoi \cite[§A.4]{Anemoi}, \ArionHash (\Cref{Tab: Arion parameters}), \Griffin  \cite[Table~2]{Griffin} and \Poseidon \cite[Table~1]{USENIX:GKRRS21} for $256$ bit prime fields and 128 bit security with $d_2 \in \{ 121, 123, 125, 161, 193, 195, 257 \}$.}
        \label{Tab: round numbers}
        \resizebox{0.7\textwidth}{!}{
            \begin{tabular}{ M{5mm} | M{1.6cm} | M{1.9cm} | M{1.4cm} | M{1.2cm} | M{2.5cm} }
                \toprule
                \multicolumn{6}{ c }{Rounds} \\
                \midrule

                & \ArionHash & \aArionHash & \Griffin & \Anemoi & \Poseidon \\
                \midrule

                $n$ & \multicolumn{5}{ c }{$d_1 = 3$} \\
                \midrule

                $3$ & $6$ & $5$ & $16$ &      & $r_f = 8,\ r_p = 84$ \\
                $4$ & $6$ & $4$ & $14$ & $12$ & $r_f = 8,\ r_p = 84$ \\
                $5$ & $4$ &     &      &      & $r_f = 8,\ r_p = 84$ \\
                $6$ & $5$ & $4$ &      & $10$ & $r_f = 8,\ r_p = 84$ \\
                $8$ & $4$ & $4$ & $11$ & $10$ & $r_f = 8,\ r_p = 84$ \\

                \midrule
                $n$ & \multicolumn{5}{ c }{$d_1 = 5$} \\
                \midrule

                $3$ & $6$ & $4$ & $12$ &      & $r_f = 8,\ r_p = 56$ \\
                $4$ & $5$ & $4$ & $11$ & $12$ & $r_f = 8,\ r_p = 56$ \\
                $5$ & $5$ & $4$ &      &      & $r_f = 8,\ r_p = 56$ \\
                $6$ & $5$ & $4$ &      & $10$ & $r_f = 8,\ r_p = 56$ \\
                $8$ & $4$ & $4$ & $9$  & $10$ & $r_f = 8,\ r_p = 56$ \\

                \bottomrule
            \end{tabular}
        }
    \end{table}
    In \Cref{Tab: R1CS constraints} we compare the theoretical number of constraints for \RoneCS of various hash functions.
    \begin{table}[H]
        \centering
        \caption{\RoneCS constraint comparison $256$ bit prime fields and 128 bit security with $d_2 \in \{ 121, 123, 125, 161, 193, 195, 257 \}$.
            Round numbers for \Anemoi, \Griffin and \Poseidon are taken from \cite[§A.4]{Anemoi}, \cite[Table~1]{Griffin} and \cite[Table~1]{USENIX:GKRRS21}.}
        \label{Tab: R1CS constraints}
        \resizebox{0.7\textwidth}{!}{
            \begin{tabular}{ M{5mm} | M{1.6cm} | M{1.9cm} | M{1.4cm} | M{1.2cm} | M{1.6cm} }
                \toprule
                \multicolumn{6}{ c  }{\RoneCS Constraints} \\
                \midrule

                & \ArionHash & \aArionHash & \Griffin & \Anemoi & \Poseidon \\
                \midrule

                $n$ & \multicolumn{5}{ c }{$d_1 = 3$} \\
                \midrule

                $3$ & $102$ & $85$  & $96$  &       & $216$ \\
                $4$ & $126$ & $84$  & $112$ & $96$  & $232$ \\
                $5$ & $120$ & $100$ &       &       & $248$ \\
                $6$ & $145$ & $116$ &       & $120$ & $264$ \\
                $8$ & $148$ & $148$ & $176$ & $160$ & $296$ \\

                \midrule
                $n$ & \multicolumn{5}{ c }{$d_1 = 5$} \\
                \midrule

                $3$ & $114$ & $76$  & $96$  &       & $240$ \\
                $4$ & $120$ & $96$  & $110$ & $120$ & $264$ \\
                $5$ & $125$ & $116$ &       &       & $288$ \\
                $6$ & $170$ & $136$ &       & $150$ & $312$ \\
                $8$ & $176$ & $176$ & $162$ & $200$ & $360$ \\

                \bottomrule
            \end{tabular}
        }
    \end{table}

    Moreover, in Appendix D.1 of the full version of the paper \cite{Arion} we compare the performance of \Arion, \Griffin and \Poseidon using the \verb!C++! library \libsnark \cite{libsnark} that is used in the privacy-protecting digital currency
    Zcash \cite{Zcash}.

    \subsection{\Plonk Performance of \ArionHash}\label{Sec: Plonk}
    \Plonk \cite{EPRINT:GabWilCio19} is a zkSNARK proof system which does not utilize \RoneCS constraints.
    In \Plonk a 2-(input)-wire constraint is of the form, see \cite[\S 6]{EPRINT:GabWilCio19},
    \begin{equation}\label{Equ: 2-wire Plonk constraint}
        (a \cdot b) \cdot q_M + a \cdot q_L + b \cdot q_R + c \cdot q_O + q_C = 0,
    \end{equation}
    $a$ and $b$ denote the left and right input variable, $c$ denotes the output variable and $q_M$, $q_L$, $q_R$, $q_O$ and $q_C$ denote the ``selector coefficient'' of the multiplication, the variables and the constant term.
    The 3-(input)-wire \Plonk constraint has 3 addition gates
    \begin{equation}\label{Equ: 3-wire Plonk constraint}
        (a \cdot b) \cdot q_M + a \cdot q_L + b \cdot q_R + c \cdot q_O + d \cdot q_F + q_C = 0,
    \end{equation}
    where $d$ is the ``fourth'' variable and $q_F$ its selector coefficient.

    Counting the number of \Plonk constraints is more subtle, since we now have to account for additions too.
    \begin{lem}
        Let $\Fp$ be a finite field, let $r, n \geq 2$ be integers, and let $\ArionHash$ with $r$ rounds be
        defined over $\Fpn$.
        For $i = 1, 2$ denote with $d_{i, \text{inc}}$ the minimal number of multiplications to compute the univariate power
        permutation $x^{d_i}$.
        \begin{enumerate}
            \item A prover circuit needs
            \begin{equation*}
                (n - 1) \cdot (d_{1, inc} + 6) + d_{2, inc} - 1
            \end{equation*}
            2-wire and
            \begin{equation*}
                (n - 1) \cdot (d_{1, inc} + 4) + d_{2, inc}
            \end{equation*}
            3-wire \Plonk constraints for the \ArionHash GTDS.

            \item A prover circuit needs
            \begin{equation*}
                4 \cdot (n - 1)
            \end{equation*}
            2-wire and
            \begin{equation*}
                \begin{dcases}
                    n, & n = 2, 3, \\
                    n + 2 + \left\lceil \frac{n - 3}{2} \right\rceil + \left\lceil \frac{n - 4}{2} \right\rceil, & n \geq 4,
                \end{dcases}
            \end{equation*}
            3-wire \Plonk constraints for the affine layer of \ArionHash.
        \end{enumerate}
        Then a prover circuit needs
        \begin{equation*}
            N_{\ArionHash, 2} =
            r \cdot \big( (n - 1) \cdot (d_{1, inc} + 6) + d_{2, inc} - 1 \big) + (r + 1) \cdot
            \begin{dcases}
                n \cdot (n - 1), & n = 2, 3 \\
                4 \cdot (n - 1), & n \geq 4,
            \end{dcases}
        \end{equation*}
        2-wire and
        \begin{align*}
            N_{\ArionHash, 3} =
            r &\cdot \big( (n - 1) \cdot (d_{1, inc} + 4) + d_{2, inc} \big) \\
            &+ (r + 1) \cdot
            \begin{dcases}
                n, &n = 2, 3, \\
                n + 2 + \left\lceil \frac{n - 3}{2} \right\rceil + \left\lceil \frac{n - 4}{2} \right\rceil, & n \geq 4,
            \end{dcases}
        \end{align*}
        3-wire \Plonk constraints for \ArionHash.
    \end{lem}
    \begin{proof}
        For (1), again we can use the CCZ-equivalent GTDS $\mathcal{G}$, see \Cref{Prop: CCZ-equivalence}, to build the circuit for the \ArionHash GTDS.
        For $x^{d_i}$ one needs $d_{i, inc}$, so we need $(n - 1) \cdot d_{1, inc} + d_{2, inc}$ constraints for the univariate permutation polynomials.
        For $\tau_{n, n}$ one needs one constraint, and for $\tau_{i + 1, n}$, where $i < n - 1$, one needs two 2-wire constraints, so in total one needs $1 + 2 \cdot (n - 2)$ 2-wire constraints to compute the $\tau_{i + 1, n}$'s.
        On the other, hand for 3-wire constraints one can compute all $\tau_{i + 1, n}$'s with one constraint, so $n - 1$ in total.
        To compute
        \begin{align*}
            g_i &= \tau^2 + \alpha_{i, 1} \cdot \tau_{i + 1, n} + \alpha_{i, 2}, \\
            h_i &= \tau^2 + \beta_i \cdot \tau_{i + 1, n}
        \end{align*}
        one needs two constraints since one can build any quadratic polynomial with the 2-wire \Plonk constraint, see \Cref{Equ: 2-wire Plonk constraint}.
        To compute $x_i^{d_1} \cdot g_i + h_i$ one needs two 2-wire constraints and one 3-wire constraint.
        We have to do this $(n - 1)$ times, hence in total we need $(n - 1) \cdot d_{1, inc} + d_{2, inc} + 1 + 2 \cdot (n - 2) + 4 \cdot (n - 1)$ 2-wire and $(n - 1) \cdot d_{1, inc} + d_{2, inc} + (n - 1) + (n - 1) \cdot (2 + 1)$ 3-wire constraints.

        For (2), we build the circuit with \Cref{Alg: efficient algorithm}.
        To compute a sum of $m$ elements one needs $m - 1$ 2-wire constraints and $1 + \left\lceil \frac{m - 3}{2} \right\rceil$ 3-wire constraints.
        We have to do this for $\sigma$ and for $\sum_{i = 2}^{n} (i - 1) \cdot v_i$, so we need $(n - 1) + (n - 2) + 1$ 2-wire and $1 + \left\lceil \frac{n - 3}{2} \right\rceil + 1 + \left\lceil \frac{n - 4}{2} \right\rceil + 1$ 3-wire constraints to compute $w_1 = \sigma + \sum_{i = 2}^{n} (i - 1) \cdot v_i + c_1$, where we do constant addition in the addition of the two sums.
        For the $i$\textsuperscript{th} component, we have that $w_i = w_{i - 1} - \sigma + n \cdot v_i - c_{i - 1} + c_i$, so we need two 2-wire and one 3-wire constraints.
        We have to do this $n - 1$ times, hence in total we need $(n - 1) + (n - 2) + 1 + 2 \cdot (n - 1)$ 2-wire and $3 + \left\lceil \frac{n - 3}{2} \right\rceil + \left\lceil \frac{n - 4}{2} \right\rceil + n - 1$ 3-wire constraints. \qed
    \end{proof}

    Note that for $n \geq 4$ \Cref{Alg: efficient algorithm} yields more efficient 2-wire and 3-wire circuits than generic matrix multiplication which always needs $n \cdot (n - 1)$ 2-wire and $n \cdot \left( 1 +  \left\lceil \frac{n - 3}{2} \right\rceil \right)$ 3-wire constraints.

    For \ArionHash's main competitors \Anemoi, \Griffin and \Poseidon we list the formulae to compute their \Plonk constraints in \Cref{Tab: Plonk formulae}.

    \begin{table}[H]
        \centering
        \caption{\Plonk constraints for \Anemoi \cite[\S 7.2]{Anemoi}, \Griffin \cite[\S 7.3]{Griffin} and \Poseidon \cite{USENIX:GKRRS21}.}
        \label{Tab: Plonk formulae}
        \resizebox{1.0\textwidth}{!}{
            \begin{tabular}{ M{1.7cm} | M{8.8cm} | M{9.5cm} }
                \toprule
                Hash & 2-wire constraints & 3-wire constraints \\
                \midrule

                \Anemoi &
                $\frac{r \cdot n}{2} \cdot (d_{inc} + 5) + (r + 1) \cdot
                \begin{cases}
                    2, & n = 2, \\
                    n \cdot \left( \frac{n}{2} - 1 \right), & n = 4, \\
                    10, & n = 6, \\
                    16, & n = 8
                \end{cases}
                $ &
                $\frac{r \cdot n}{2} \cdot (d_{inc} + 3) + (r + 1) \cdot
                \begin{cases}
                    (r + 1) \cdot n, & n = 2, 4, \\
                    6, & n = 6, \\
                    12, & n = 8
                \end{cases}
                $ \\[5pt]

                \hline

                \Griffin &
                $r \cdot (2 \cdot d_{inc} + 4 \cdot n - 11) + (r + 1) \cdot
                \begin{cases}
                    5, & n = 3, \\
                    8, & n = 4, \\
                    24, & n = 8, \\
                    \frac{8 \cdot n}{4} + 2 \cdot n - 4, & n \geq 12
                \end{cases}
                $ &
                $r \cdot (2 \cdot d_{inc} + 3 \cdot n - 8) + (r + 1) \cdot
                \begin{cases}
                    3, & n = 3, \\
                    6, & n = 4, \\
                    20, & n = 8, \\
                    \frac{6 \cdot n}{4} + 4 \cdot \left\lfloor \frac{\frac{n}{4} - 1}{2} \right\rfloor + n, & n \geq 12
                \end{cases}
                $\\[25pt] 

                \hline

                \Poseidon &
                $d_{inc} \cdot \left( n \cdot r_f + r_p \right) + (r + 1) \cdot n \cdot (n - 1)$ &
                $d_{inc} \cdot \left( n \cdot r_f + r_p \right) + (r + 1) \cdot n \cdot
                \begin{cases}
                    n, & n = 2,3, \\
                    \left\lceil \frac{n - 3}{2} \right\rceil, & n \geq 4
                \end{cases}
                $ \\

                \bottomrule
            \end{tabular}
        }
    \end{table}

    In \Cref{Tab: Plonk constraints} we compare the theoretical number of constraints for \Plonk for various hash functions.
    \clearpage 
    \begin{table}[H]
        \centering
        \caption{\Plonk constraint comparison $256$ bit prime fields and 128 bit security with $d_2 \in \{ 121, 123, 125, 161, 193, 195, 257 \}$.
            Round numbers are the same as in \Cref{Tab: round numbers}.}
        \label{Tab: Plonk constraints}
        \resizebox{0.8\textwidth}{!}{
            \begin{tabular}{ M{1.9cm} | M{8mm} M{8mm} M{8mm} M{8mm} M{8mm} | M{8mm} M{8mm} M{8mm} M{8mm} M{8mm} }
                \toprule
                \multicolumn{11}{ c }{\Plonk Constraints} \\
                \midrule

                & \multicolumn{5}{ c | }{2-wire constraints} & \multicolumn{5}{ c }{3-wire constraints} \\
                \cmidrule{2-11}

                & \multicolumn{10}{ c }{State size $n$} \\
                \midrule

                & $3$ & $4$ & $5$ & $6$ & $8$ & $3$ & $4$ & $5$ & $6$ & $8$ \\
                \midrule

                Hash & \multicolumn{10}{c}{$d_1 = 3$} \\
                \midrule

                \ArionHash  & $200$ & $276$  & $296$  & $360$  & $396$  & $147$ & $211$ & $219$ & $261$  & $279$  \\
                \aArionHash & $168$ & $188$  & $240$  & $292$  & $396$  & $123$ & $143$ & $177$ & $211$  & $279$  \\
                \Poseidon   & $768$ & $1336$ & $2088$ & $3024$ & $5448$ & $492$ & $600$ & $708$ & $1368$ & $2504$ \\
                \Griffin    & $165$ & $246$  &        &        & $563$  & $131$ & $202$ &       &        & $460$  \\
                \Anemoi     &       & $220$  &        & $320$  & $456$  &       & $172$ &       & $216$  & $332$  \\

                \midrule
                Hash & \multicolumn{10}{c}{$d_1 = 5$} \\
                \midrule

                \ArionHash  & $212$ & $247$  & $316$  & $385$  & $424$  & $159$ & $192$ & $239$ & $286$  & $307$  \\
                \aArionHash & $144$ & $200$  & $256$  & $312$  & $424$  & $107$ & $155$ & $193$ & $231$  & $307$  \\
                \Poseidon   & $624$ & $1032$ & $1568$ & $2232$ & $3944$ & $432$ & $520$ & $608$ & $1080$ & $1896$ \\
                \Griffin    & $173$ & $275$  &        &        & $561$  & $123$ & $182$ &       &        & $398$  \\
                \Anemoi     &       & $244$  &        & $350$  & $496$  &       & $196$ &       & $246$  & $372$  \\

                \bottomrule
            \end{tabular}
        }
    \end{table}

    Moreover, in Appendix D.2 of the full version of the paper \cite{Arion} we compare the performance of \Arion and \Poseidon using the \verb|Rust| library Dusk Network \Plonk \cite{Dusk-Plonk}.

    \subsubsection*{Acknowledgments.}
    Matthias Steiner and Stefano Trevisani were supported by the KWF under project number KWF-3520|31870|45842.

    \bibliographystyle{splncs04.bst}
    \bibliography{abbrev0.bib,crypto.bib,literature.bib}

    \begin{appendix}
    \section{Statistical Attacks on \Arion}
    \subsection{Differential Cryptanalysis}\label{Sec: differential cryptanalysis}
    Differential cryptanalysis \cite{JC:BihSha91} and its variants are the most widely applied attack vectors against symmetric-key ciphers and hash functions.
    It is based on the propagation of input differences through the rounds of a block cipher.
    In its base form an attacker requests the ciphertexts for large numbers of chosen plaintexts.
    Then he assumes that for $r - 1$ rounds the input difference is $\boldsymbol{\Delta x} \in \Fqn \setminus \{ \mathbf{0} \}$ and the output difference is $\boldsymbol{\Delta y} \in \Fqn$.
    Under the assumption that the differences in the last round are fixed the attacker can then deduce the possible keys.
    The key quantity to estimate the effectiveness of differential cryptanalysis is the so-called differential uniformity.
    \begin{defn}[{see \cite{EC:Nyberg93}}]
        Let $\Fq$ be a finite field, and let $f: \Fqn \to \Fqm$ be a function.
        \begin{enumerate}
            \item The differential distribution table of $f$ at $\mathbf{a} \in \Fqn$ and $\mathbf{b} \in \Fqm$ is defined as
            \begin{equation*}
                \delta_f (\mathbf{a}, \mathbf{b}) = \left\vert \{ \mathbf{x} \in \Fqn \mid f (\mathbf{x} + \mathbf{a}) - f (\mathbf{x}) = f (\mathbf{b}) \} \right\vert.
            \end{equation*}

            \item The differential uniformity of $f$ is defined as
            \begin{equation*}
                \delta (f) = \max_{ \substack{\mathbf{a} \in \Fqn \setminus \{ \mathbf{0} \},\\ \mathbf{b} \in \Fqm} } \delta_f (\mathbf{a}, \mathbf{b}).
            \end{equation*}
        \end{enumerate}
    \end{defn}

    Given the differential uniformity of a function one can upper bound the success probability of differential cryptanalysis with input differences $\boldsymbol{\Delta x} \in \Fqn \setminus \{ \mathbf{0} \}$ and $\boldsymbol{\Delta y} \in \Fqm$ by
    \begin{equation}
        \prob \left[ f \!: \boldsymbol{\Delta x} \to \boldsymbol{\Delta y} \right] \leq \frac{\delta (f)}{q^n}.
    \end{equation}
    Naturally, the lower the differential uniformity the stronger is the resistance of a block cipher against differential cryptanalysis.

    By the choice of our parameters of \Arion, see \Cref{Def: GTDS} and thereafter, we have that $d_2 \geq d_1$.
    With \cite[Theorem~18, Corollary~19]{GTDS} the maximal success probability of any differential for the \Arion GTDS is bounded by
    \begin{equation}\label{Equ: differential uniformity bound}
        \prob \left[ \mathcal{F}_\Arion \!: \boldsymbol{\Delta x} \to \boldsymbol{\Delta y} \right] \leq \left( \frac{d_2}{p} \right)^{\wt \left( \boldsymbol{\Delta x} \right)} \leq \frac{d_2}{p},
    \end{equation}
    where $\boldsymbol{\Delta x} \in \Fpn \setminus \{ \mathbf{0} \}$, $\boldsymbol{\Delta y} \in \Fpn$, and $1 \leq \wt \left( \boldsymbol{\Delta x} \right) \leq n$ denotes the Hamming weight, i.e.\ the number of non-zero entries, of $\boldsymbol{\Delta x}$.
    Since $\mathcal{R}_\mathbf{k}^{(i)}$ and $\mathcal{F}_\Arion^{(i)}$ are affine equivalent \Cref{Equ: differential uniformity bound} also applies to $\mathcal{R}_\mathbf{k}^{(i)}$.
    For \Arion we use a differential in every round.
    If we assume that the differentials are independent among the rounds of \Arion, then
    \begin{equation}\label{Equ: differential trail}
        \prob \left[ \mathcal{R}_\mathbf{k}^{(r)} \circ \cdots \circ \mathcal{R}_\mathbf{k}^{(1)} \! : \boldsymbol{\Delta x}_1 \to \ldots \to \boldsymbol{\Delta x}_{r + 1} \right] \leq \left( \frac{d_2}{p} \right)^{ \wt (\boldsymbol{\Delta x}_1) + \ldots + \wt (\boldsymbol{\Delta x}_r)}.
    \end{equation}

    Therefore, we can estimate the security level $\kappa$ of \Arion against any differential trail via
    \begin{equation}\label{Equ: security level differential cryptanalysis single trail}
        \left( \frac{d_2}{p} \right)^r \leq 2^{-\kappa} \Longrightarrow \kappa \leq r \cdot \Big( \log_2 \left( p \right) - \log_2 \left( d_2 \right) \Big).
    \end{equation}
    In \Cref{Tab: security level differential cryptanalysis single characteristic} we list the security level of \Arion against a differential characteristic for different prime sizes.
    \begin{table}[H]
        \centering
        \caption{Security level of \Arion against any differential characteristic for $p \geq 2^N$ and $d_2 \leq 2^9$.}
        \label{Tab: security level differential cryptanalysis single characteristic}
        \begin{tabular}{ M{0.8cm} | M{0.8cm} | M{1.5cm} }
            \toprule
            $r$ & $N$ & $\kappa$ (bits) \\
            \midrule

            $3$ & $60$  & $153$ \\
            $2$ & $120$ & $222$ \\
            $1$ & $250$ & $241$ \\

            \bottomrule
        \end{tabular}
    \end{table}

    With \Cref{Equ: differential trail} we can also estimate the probability of the differential hull of \Arion.
    \begin{thm}[{Differential hull of \Arion}]\label{Th: differential hull}
        Let $p \in \mathbb{Z}$ be a prime and let $\Fp$ be the field with $p$ elements, and let $n > 1$ and $r \geq 1$ be integers.
        Let $\mathcal{R}_\mathbf{k}^{(1)}, \dots, \mathcal{R}_\mathbf{k}^{(r)}: \Fpn \times \Fpn \to \Fpn$ be \Arion round functions, and let $\boldsymbol{\Delta x}_1, \boldsymbol{\Delta x}_{r + 1} \in \Fpn$ be such that $\boldsymbol{\Delta x}_1 \neq \mathbf{0}$.
        Assume that the differentials among the rounds of \Arion are independent.
        Then
        \begin{equation*}
            \prob \left[ \mathcal{R}_\mathbf{k}^{(r)} \circ \cdots \circ \mathcal{R}_\mathbf{k}^{(1)} \!: \boldsymbol{\Delta x}_1 \to \boldsymbol{\Delta x}_{r + 1} \right] \leq \left( \frac{d_2}{p} \right)^{\wt (\boldsymbol{\Delta x}_1)} \cdot \Big( \left( d_2 + 1 \right)^n - 1 \Big)^{r - 1}.
        \end{equation*}
    \end{thm}
    \begin{proof}
        Let us first do some elementary rearrangements and estimations
        \begin{align*}
            &\prob \left[ \mathcal{R}_\mathbf{k}^{(r)} \circ \cdots \circ \mathcal{R}_\mathbf{k}^{(1)} \!: \boldsymbol{\Delta x}_1 \to \boldsymbol{\Delta x}_r \right] \\
            &= \sum_{\boldsymbol{\Delta x}_2, \dots, \boldsymbol{\Delta x}_r \in \Fpn \setminus \{ \mathbf{0} \}} \prob \left[ \bigcap_{i = 1}^{r} \left\{ \mathcal{R}_\mathbf{k}^{(i)} \!: \boldsymbol{\Delta x}_i \to \boldsymbol{\Delta x}_{i + 1} \right\} \right] \\
            &\stackrel{(1)}{=} \sum_{\boldsymbol{\Delta x}_2, \dots, \boldsymbol{\Delta x}_r \in \Fpn \setminus \{ \mathbf{0} \}} \prob \left[ \bigcap_{i = 1}^{r} \left\{ \mathcal{F}_\Arion^{(i)} \!: \boldsymbol{\Delta x}_i \to \circulant (1, \dots, n)^{-1} \boldsymbol{\Delta x}_{i + 1} \right\} \right] \\
            &\stackrel{(2)}{=} \sum_{\boldsymbol{\Delta x}_2, \dots, \boldsymbol{\Delta x}_r \in \Fpn \setminus \{ \mathbf{0} \}} \prob \left[ \bigcap_{i = 1}^{r} \left\{ \mathcal{F}_\Arion^{(i)} \!: \boldsymbol{\Delta x}_i \to \boldsymbol{\Delta x}_{i + 1} \right\} \right] \\
            &\stackrel{(3)}{\leq} \left( \frac{d_2}{p} \right)^{\wt (\boldsymbol{\Delta x}_1)} \cdot \sum_{\boldsymbol{\Delta x}_2, \dots, \boldsymbol{\Delta x}_r \in \Fpn \setminus \{ \mathbf{0} \}} \left( \frac{d_2}{p} \right)^{\wt (\boldsymbol{\Delta x}_2) + \ldots + \wt (\boldsymbol{\Delta x}_r)}.
        \end{align*}
        In $(1)$ we expanded the definition of the round function and inverted the circulant matrix, in $(2)$ we exploited that we sum over all possible differentials and that the affine \Arion layer is invertible, moreover we implicitly substituted $\boldsymbol{\Delta x}_{r + 1} = \circulant (1, \dots, n)^{-1} \boldsymbol{\Delta x}_{r + 1}$.
        Finally, in $(3)$ we applied \Cref{Equ: differential trail}.
        For ease of writing we compute the sum only for one difference variable, then
        \begin{align*}
            &\sum_{\boldsymbol{\Delta x} \in \Fpn \setminus \{ \mathbf{0} \}} \left( \frac{d_2}{p} \right)^{\wt (\boldsymbol{\Delta x})}
            \stackrel{(4)}{=} \sum_{i = 1}^{n} (p - 1)^i \cdot \binom{n}{i} \cdot \left( \frac{d_2}{p} \right)^i \\
            &\leq \sum_{i = 1}^{n} \binom{n}{i} \cdot d_2^i = \left( d_2 + 1 \right)^n - 1,
        \end{align*}
        where in $(4)$ we used that there are $(p - 1)^i \cdot \binom{n}{i}$ many vectors $\boldsymbol{\Delta x} \in \Fpn$ with $\wt (\boldsymbol{\Delta x}) = i$.
        This proves the claim. \qed
    \end{proof}
    \begin{rem}
        This estimation can be performed over any finite field $\Fq$, any invertible affine layer, and any primitive whose differential uniformity at round level is in $\mathcal{O} \left( p^{-\wt (\boldsymbol{\Delta x})} \right)$.
    \end{rem}

    With the theorem we can now estimate the security level of \Arion with respect to differential cryptanalysis and the full differential hull via
    \begin{align}
        &\left( \frac{d_2}{p} \right)^{\wt (\boldsymbol{\Delta x}_1)} \cdot \Big( \left( d_2 + 1 \right)^n - 1 \Big)^{r - 1} \leq 2^{- \kappa} \\
        &\Longrightarrow \kappa \leq \wt (\boldsymbol{\Delta x}_1) \cdot \Big( \log_2 \left( p \right) - \log_2 \left( d_2 \right) \Big) - \left( r - 1 \right) \cdot \log_2 \Big( \left( d_2 + 1 \right)^n - 1 \Big).
    \end{align}
    In \Cref{Tab: security level differential cryptanalysis} report the security level for \Arion against differential cryptanalysis utilizing the full differential hull with the parameters from \Cref{Tab: Arion parameters} and different field sizes.
    Since our probability estimation from \Cref{Th: differential hull} could in principle be $> 1$ for some parameters combinations, we always report the security level with respect to the smallest $\wt (\boldsymbol{\Delta x}_1)$ such that the probability estimate is $< 1$.

    \begin{table}[H]
        \centering
        \caption{Security level of \Arion against any differential characteristic for $p \geq 2^N$ and $d_2 \leq 257$ with the full differential hull.}
        \label{Tab: security level differential cryptanalysis}
        \resizebox{1.0\textwidth}{!}{
            \begin{tabular}{ M{1.5cm} | M{0.6cm} M{0.6cm} M{0.6cm} M{0.6cm} M{0.6cm} M{0.6cm} | M{0.6cm} M{0.6cm} M{0.6cm} M{0.6cm} M{0.6cm} M{0.6cm} | M{0.6cm} M{0.6cm} M{0.6cm} M{0.6cm} M{0.6cm} M{0.6cm} }
                \toprule
                & \multicolumn{6}{ c | }{$N = 60$} & \multicolumn{6}{ c | }{$N = 120$} & \multicolumn{6}{ c }{$N = 250$} \\
                \midrule

                $n$                             & $3$   & $4$   & $4$  & $5$  & $6$  & $8$  
                & $3$   & $4$   & $4$  & $5$  & $6$  & $8$  
                & $3$   & $4$   & $4$  & $5$  & $6$  & $8$  
                \\
                $r$                             & $6$   & $5$   & $6$  & $5$  & $5$  & $4$  
                & $6$   & $5$   & $6$  & $5$  & $5$  & $4$  
                & $6$   & $5$   & $6$  & $5$  & $5$  & $4$  
                \\
                $\wt (\boldsymbol{\Delta x}_1)$ & $3$   & $3$   & $4$  & $4$  & $4$  & $4$  
                & $2$   & $2$   & $2$  & $2$  & $1$  & $2$  
                & $1$   & $1$   & $1$  & $1$  & $1$  & $1$  
                \\
                $\kappa$ (bits)                 & $35$  & $27$  & $47$ & $47$ & $15$ & $15$ 
                & $103$ & $95$  & $63$ & $63$ & $31$ & $31$ 
                & $121$ & $113$ & $81$ & $81$ & $49$ & $49$ 
                \\

                \bottomrule
            \end{tabular}
        }
    \end{table}

            %
            %
            %
            %
            %
            %
            %

    Although, our target security of $128$ bit is never met we should keep in mind that the full differential hull is of size $(p - 1)^{n \cdot (r - 1)}$.
    E.g., for our target primes $\text{BLS}12$ and $\text{BN}254$ the smallest differential hull is of size $\approx 2^{250 \cdot 15}$, and for $60$ bit prime fields the smallest differential hull still would be of size $2^{60 \cdot 15}$.
    Therefore, we do not expect that differential cryptanalysis can break \Arion \& \ArionHash within the $128$ bit security target.

    Nevertheless, to convince skeptical readers we provide another hull estimation for a computationally limited adversary, i.e.\ an adversary who can only search a restricted number of differences within the differential hull.
    \begin{lem}[{Restricted differential hull of \Arion}]
        Let $p \in \mathbb{Z}$ be a prime and let $\Fp$ be the field with $p$ elements, and let $n > 1$ and $r \geq 1$ be integers.
        Let $\mathcal{R}_\mathbf{k}^{(1)}, \dots, \mathcal{R}_\mathbf{k}^{(r)}: \Fpn \times \Fpn \to \Fpn$ be \Arion round functions, and let $\boldsymbol{\Delta x}_1, \boldsymbol{\Delta x}_{r + 1} \in \Fpn$ be such that
        $\boldsymbol{\Delta x}_1 \neq \mathbf{0}$.
        Assume that the differentials among the rounds of \Arion are independent and that in every round only up to $M < (p - 1) \cdot n$ many differences $\boldsymbol{\Delta x}_i \in \Fpn \setminus \{ \mathbf{0} \}$ can be utilized.
        Then
        \begin{equation*}
            \prob \left[ \mathcal{R}_\mathbf{k}^{(r)} \circ \cdots \circ \mathcal{R}_\mathbf{k}^{(1)} \!: \boldsymbol{\Delta x}_1 \to \boldsymbol{\Delta x}_{r + 1} \right] \leq \left( \frac{d_2}{p} \right)^{\wt (\boldsymbol{\Delta x}_1)} \cdot \left( \frac{M \cdot d_2}{p} \right)^{r - 1}.
        \end{equation*}
    \end{lem}
    \begin{proof}
        Recall that the bound for the differential uniformity of the \Arion GTDS is maximal if there is only one non-zero entry in $\boldsymbol{\Delta x} \in \Fpn \setminus \{ \mathbf{0} \}$, see  \Cref{Equ: differential uniformity bound}.
        In every intermediate round there are $(p - 1) \cdot n$ many differences with one non-zero entry, hence to maximize the estimation we first sum over those elements.
        With \Cref{Equ: differential trail} we then obtain the claimed inequality. \qed
    \end{proof}

    Provided that $M < \frac{p}{d_2}$, then the inequality is always less than $1$.
    For the restricted differential hull we can then estimate the security level of \Arion via
    \begin{align}
        &M^{r - 1} \cdot \left( \frac{d_2}{p} \right)^r \leq 2^{-\kappa} \\
        &\Longrightarrow \kappa \leq r \cdot \Big( \log_2 \left( p \right) - \log_2 \left( d_2 \right) \Big) - \left( r - 1 \right) \cdot \log_2 \left( M \right).
    \end{align}
    In \Cref{Tab: security level differential cryptanalysis restricted hull} we report round number for various field sizes that achieve at least $128$ bit security for a restricted hull of size $\sqrt{p}$ at round level.

    \begin{table}[H]
        \centering
        \caption{Security level of \Arion against differential cryptanalysis for $p \geq 2^N$ and $d_2 \leq 257$ with a restricted differential hull of size $2^{N / 2}$ at round level.}
        \label{Tab: security level differential cryptanalysis restricted hull}
        \begin{tabular}{ M{0.8cm} | M{0.8cm} | M{1.5cm} }
            \toprule
            $r$ & $N$   & $\kappa$ (bits) \\
            \midrule

            $5$ & $60$  & $139$ \\
            $4$ & $120$ & $267$ \\
            $3$ & $250$ & $475$ \\

            \bottomrule
        \end{tabular}
    \end{table}

    \subsection{Truncated Differential and Rebound Attacks}\label{Sec: truncated differential cryptanalysis}
    In a truncated differential attack \cite{FSE:Knudsen94} an attacker can only predict parts of the difference between pairs of text.
    We expect that the \Arion GTDS admits truncated differentials with probability $1$.
    For example $\left( 0, \dots, 0, \alpha \right)^\intercal \xrightarrow{\mathcal{F}_\Arion} \circulant \left( 1, \dots, n \right) \left( 0, \dots, 0, \beta \right)^\intercal$ where $\alpha, \beta \in \Fpx$.
    Note that for any vector $\mathbf{v} \in \Fpn$ with $\wt (\mathbf{v})$ one always has that $\wt \big( \circulant \left(1, \dots, n \right) \mathbf{v} \big) = n$.
    Therefore, for any truncated differential of the \Arion GTDS where only one element is active in the input and output one has that all components are active after application of the affine layer of \Arion.
    For truncated differentials we then estimate the security level of \Arion via
    \begin{equation}
        M \cdot \left( \frac{d_2}{p} \right)^n \leq 2^{-\kappa} \Longrightarrow \kappa \leq n \cdot \left( \log_2 \left( p \right) - \log_2 \left( d_2 \right) \right) - \log_2 \left( M \right),
    \end{equation}
    where $M$ denotes the size of the restricted differential hull available to a computationally limited adversary.
    In \Cref{Tab: security level truncated differential cryptanalysis of weight 1} we report the security level for various parameters.
    For primes of size $p \geq 2^{120}$ one full round is already sufficient to achieve $128$ bit security for a restricted differential hull of size $M \geq 2^{120}$.
    On the other hand for $p \geq 2^{60}$ the security target is not met.
    Therefore, if one would like to instantiate \Arion \& \ArionHash over such a prime one also has to consider later rounds.
    \clearpage 
    \begin{table}[H]
        \centering
        \caption{Security level of \Arion against truncated differential cryptanalysis with weight $1$ truncated differentials in the first round and for $p \geq 2^N$ and $d_2 \leq 257$ with a restricted differential hull of size $2^{M}$.}
        \label{Tab: security level truncated differential cryptanalysis of weight 1}
        \begin{tabular}{ M{0.8cm} | M{0.8cm} | M{0.8cm} | M{1.5cm} }
            \toprule
            $N$ & $n$ & $M$ & $\kappa$ (bits) \\
            \midrule

            $60$  & $4$ & $100$ & $107$ \\
            $120$ & $3$ & $120$ & $215$ \\
            $250$ & $3$ & $250$ & $475$ \\

            \bottomrule
        \end{tabular}
    \end{table}

    Suppose an adversary can cover two rounds with a truncated differential of probability $1$, his best bet is to search the remaining rounds for input/output differentials of weight $1$.
    For such an adversary we can then estimate the security level as
    \begin{equation}
        \left( M \cdot \frac{d_2}{p} \right)^{r - 2} \leq 2^{-\kappa} \Longrightarrow \kappa \leq \left( r - 2 \right) \cdot \left( \log_2 \left( p \right) - \log_2 \left( d_2 \right) - \log_2 \left( M \right) \right),
    \end{equation}
    where $M$ denotes the size of the restricted differential hull available to a computationally limited adversary at round level.
    In \Cref{Tab: two round truncated differential cryptanalysis} we report the security level for various parameters.
    For primes of size $p \geq 2^{250}$ four rounds of \Arion are sufficient to protect against two round truncated differentials of probability $1$ within the $128$ bit security target.
    On the other hand, for primes of size $p \geq 2^{60}$ one needs at least $6$ rounds to meet the $128$ bit security target.
    \begin{table}[H]
        \centering
        \caption{Security level of \Arion against two round truncated differentials for $p \geq 2^N$ and $d_2 \leq 257$ with a restricted differential hull of size $2^{N / 2}$ at round level.}
        \label{Tab: two round truncated differential cryptanalysis}
        \begin{tabular}{ M{0.8cm} | M{0.8cm} | M{1.5cm} }
            \toprule
            $N$ & $r - 2$ & $\kappa$ (bits) \\
            \midrule

            $60$  & $4$ & $87$  \\
            $120$ & $3$ & $155$ \\
            $250$ & $2$ & $233$ \\

            \bottomrule
        \end{tabular}
    \end{table}

    In a rebound attack \cite{AC:LMRRS09,FSE:MRST09} one has to find two input/output pairs such that the inputs satisfy a certain (truncated) input difference and the outputs satisfy a certain (truncated) output difference.
    Such an attack can be split into two phases: an \textit{inbound} and an \textit{outbound} phase.
    Let $P_\Arion: \Fqn \to \Fqn$ be the target permutation, then we split it into three sub-parts $P_\Arion = P_{fw} \circ P_{in} \circ P_{out}$.
    The inbound phase is placed in the middle preceded and followed by the two outbound phases.
    Then, in the outbound phase two high-probability (truncated) differential trails are constructed which are connected with the inbound phase.
    Since a truncated differential with probability $1$ can only cover a single round an attacker can cover only $r - 2$ rounds with an inside-out approach.
    By \Cref{Tab: two round truncated differential cryptanalysis}, for our target primes $\text{BLS}12$ and $\text{BN}254$ two inbound rounds are sufficient to achieve the $128$ bit security target.
    So in total $4$ rounds are sufficient to nullify this attack vector.

    \subsection{Linear Cryptanalysis}\label{Sec: linear cryptanalysis}
    In linear cryptanalysis \cite{EC:Matsui93,SAC:BaiSteVau07} one tries to discover affine approximations of round functions for a sample of known plaintexts.
    The key quantity to estimate the effectiveness of linear cryptanalysis is the so-called correlation.
    We denote with $\braket{\mathbf{a}, \mathbf{b}} = \sum_{i = 1}^{n} a_i \cdot b_i$ the scalar product over $\Fqn$.
    \begin{defn}[{see \cite[Definition~6, 15]{SAC:BaiSteVau07}}]\label{Def: correlation}
        Let $\Fq$ be a finite field, let $n \geq 1$, let $\chi: \Fq \to \mathbb{C}$ be a non-trivial additive character, let $F: \Fqn \to \Fqn$ be a function, and let $\mathbf{a}, \mathbf{b} \in \Fqn$.
        \begin{enumerate}
            \item\label{Item: correlation} The correlation for the character $\chi$ of the linear approximation $(\mathbf{a}, \mathbf{b})$ of $F$ is defined as
            \begin{equation*}
                \CORR_F (\chi, \mathbf{a}, \mathbf{b}) = \frac{1}{q^n} \cdot \sum_{\mathbf{x} \in \Fqn} \chi \Big( \big< \mathbf{a}, F (\mathbf{x}) \big> + \braket{\mathbf{b}, \mathbf{x}} \! \Big).
            \end{equation*}

            \item The linear probability for the character $\chi$ of the linear approximation $(\mathbf{a}, \mathbf{b})$ of $F$ is defined as
            \begin{equation*}
                \LP_F (\chi, \mathbf{a}, \mathbf{b}) = \left| \CORR_F (\chi, \mathbf{a}, \mathbf{b})  \right|^2.
            \end{equation*}
        \end{enumerate}
    \end{defn}
    \begin{rem}
        To be precise Baign\`eres et al.\ \cite{SAC:BaiSteVau07} defined linear cryptanalysis over arbitrary abelian groups, in particular for maximal generality they defined the correlation with respect to two additive characters $\chi, \psi: \Fq \to \mathbb{C}$ as
        \begin{equation}\label{Equ: original correlation definition}
            \CORR_F (\chi, \psi, \mathbf{a}, \mathbf{b}) = \frac{1}{q^n} \cdot \sum_{\mathbf{x} \in \Fqn} \chi \Big( \big< \mathbf{a}, F (\mathbf{x}) \big> \Big) \cdot \psi \Big( \big< \mathbf{b}, \mathbf{x} \big> \! \Big).
        \end{equation}
        Let $\Fq$ be a finite field of characteristic $p$, and let $\Tr: \Fq \to \Fp$ be the absolute trace function, see \cite[2.22.~Definition]{Niederreiter-FiniteFields}.
        For all $x \in \Fq$ we define the function $\chi_1$ as
        \begin{equation*}
            \chi_1 (x) = \exp \left(\frac{2 \cdot \pi \cdot i}{p} \cdot \Tr (x) \right).
        \end{equation*}
        Then for every non-trivial additive character $\chi: \Fq \to \mathbb{C}$ there exist $a \in \Fq^\times$ such that $\chi (x) = \chi_1 (a \cdot x)$, see \cite[5.7.~Theorem]{Niederreiter-FiniteFields}.
        Therefore, after an appropriate rescaling that we either absorb into $\mathbf{a}$ or $\mathbf{b}$ we can transform \Cref{Equ: original correlation definition} into \Cref{Def: correlation} \ref{Item: correlation}.
    \end{rem}

    For \Arion we approximate every round by affine functions, and we call the tuple $\Omega = ( \mathbf{a}_0, \dots, \mathbf{a}_r) \subset \left( \Fpn \right)^{r + 1}$ a linear trail for \Arion, where $(\mathbf{a}_{i - 1}, \mathbf{a}_i)$ is the affine approximation of the $i$\textsuperscript{th} round.
    Note that for $\mathbf{A} \in \Fpnxn$ and $\mathbf{c} \in \Fpn$ one has that $\LP_{\mathbf{A} F + \mathbf{c}} (\chi, \mathbf{a}, \mathbf{b}) = \LP_F (\chi, \mathbf{A}^\intercal \mathbf{a}, \mathbf{b})$.
    I.e., to bound the correlation of a single round it suffices to bound the correlation of the \Arion GTDS.
    By \cite[Theorem~24, Corollary~25]{GTDS} we have that
    \begin{equation}\label{Equ: Arion GTDS linear probability}
        \LP_{\mathcal{F}_\Arion} (\chi, \mathbf{a}, \mathbf{b}) \leq
        \begin{dcases}
            1, & \mathbf{a}, \mathbf{b} = \mathbf{0}, \\
            \frac{\left( d_2 - 1 \right)^2}{q}, & \text{else}.
        \end{dcases}
    \end{equation}
    If we assume that the approximations among the rounds of \Arion are independent, then
    \begin{equation}
        \LP_{\Arion} (\chi, \Omega) \leq \left( \frac{\left( d_2 - 1 \right)^2}{q} \right)^r.
    \end{equation}
    If a distinguisher is limited to $D$ queries, then under heuristic assumptions Baign\`eres et al.\ proved \cite[Theorem~7]{SAC:BaiSteVau07} that the advantage of a distinguisher, which we call success probability, for a single linear trail is lower bounded by
    \begin{equation}
        p_{success} \succeq  1 - e^{-\frac{D}{4} \cdot \LP_{\Arion} (\chi, \Omega)} \leq 1 - e^{-\frac{D}{4} \cdot \left( \frac{\left( d_2 - 1 \right)^2}{q} \right)^r}.
    \end{equation}
    For $p \geq 2^N$ and $d_2 \leq 2^M$ we approximate $\exp (-x) \approx 1 - x$ and estimate the security level $\kappa$ of \Arion against a linear trail $\Omega$ as
    \begin{align}
        p_{success} &\approx \frac{D}{4} \cdot \left( \frac{\left( d_2 - 1 \right)^2}{q} \right)^r \leq \frac{D}{4} \cdot 2^{r \cdot (2 \cdot M - N)} \leq 2^{-\kappa} \label{Equ: security level linear cryptanalysis} \\
        \Longrightarrow \kappa &\leq 2 + r \cdot \left( N - 2 \cdot M \right) - \log_2 \left( D \right).
    \end{align}

    Now recall that \Arion is supposed to be instantiated with a $64$ bit prime number $p \gtrsim 2^{64}$ and that $d_2 \in \{ 121, 123, 125, 129, 161, 193, 195, 257 \}$, then \Cref{Equ: security level linear cryptanalysis} implies the following security levels for \Arion.
    \begin{table}[H]
        \centering
        \caption{Security level of \Arion against any linear trail for $p \geq 2^N$, $d_2 \leq 2^9$ and data amount $2^M$.}
        \label{Tab: security level linear cryptanalysis}
        \begin{tabular}{ c | c  c  c | c c | c c }
            \toprule
            & \multicolumn{3}{ c | }{$N = 60$} & \multicolumn{2}{ c | }{$N = 120$} & \multicolumn{2}{ c }{$N = 250$} \\
            \midrule

            \phantom{x}$r$\phantom{x}             & $5$   & $6$   & $8$   & $3$   & $4$   & $2$   & $3$ \\
            \phantom{x}$M$\phantom{x}             & $60$  & $120$ & $180$ & $120$ & $240$ & $250$ & $500$ \\
            \phantom{x}$\kappa$ (bits)\phantom{x} & \phantom{x}$152$\phantom{x} & \phantom{x}$134$\phantom{x} & \phantom{x}$158$\phantom{x} & \phantom{x}$188$\phantom{x} & \phantom{x}$170$\phantom{x} & \phantom{x}$216$\phantom{x} & \phantom{x}$198$\phantom{x} \\

            \bottomrule
        \end{tabular}
    \end{table}

    For a computationally limited adversary we next derive an estimation of the probability of a restricted linear hull.
    \begin{lem}[{Restricted linear hull of \Arion}]
        Let $p \in \mathbb{Z}$ be a prime and let $\Fp$ be the field with $p$ elements, let $n > 1$ and $r \geq 1$ be integers, and let $\chi: \Fp \to \mathbb{C}$ be a non-trivial additive character.
        Let $\mathcal{R}_\mathbf{k}^{(1)}, \dots, \mathcal{R}_\mathbf{k}^{(r)}: \Fpn \times \Fpn \to \Fpn$ be \Arion round functions, and let $(\mathbf{a}_1, \mathbf{a}_r) \in \left( \Fpn \setminus \{ \mathbf{0} \} \right)^2 \in \left( \Fpn \setminus \{ \mathbf{0} \} \right)^2$.
        Assume that the linear approximations among the rounds of \Arion are independent and that in every round only up to $M < (p - 1) \cdot n$ many approximations $(\mathbf{a}_i, \mathbf{a}_{i + 1}) \in \left( \Fpn \setminus \{ \mathbf{0} \} \right)^2$ can be utilized.
        Then
        \begin{equation*}
            \LP_{\mathcal{R}_\mathbf{k}^{(r)} \circ \cdots \circ \mathcal{R}_\mathbf{k}^{(1)}} (\chi, \mathbf{a}_1, \mathbf{a}_r) \leq M^{r - 1} \cdot \left( \frac{\left( d_2 - 1 \right)^2}{p} \right)^r.
        \end{equation*}
    \end{lem}
    \begin{proof}
        By the independence of rounds we have that
        \begin{align*}
            \LP_{\mathcal{R}_\mathbf{k}^{(r)} \circ \cdots \circ \mathcal{R}_\mathbf{k}^{(1)}}
            &= \sum_{\mathbf{a}_2, \dots, \mathbf{a}_{r - 1}} \prod_{i = 1}^{r} \ \LP_{\mathcal{R}_\mathbf{k}^{(i)}} (\chi, \mathbf{a}_i, \mathbf{a}_{i + 1}) \\
            &= \sum_{\mathbf{a}_2, \dots, \mathbf{a}_{r - 1}} \prod_{i = 1}^{r} \ \LP_{\mathcal{F}^{(i)}} \left( \chi, \circulant (1, \dots, n)^\intercal \mathbf{a}_i, \mathbf{a}_{i + 1} \right) \\
            &\leq \sum_{\mathbf{a}_2, \dots, \mathbf{a}_{r - 1}} \left( \frac{\left( d_2 - 1 \right)^2}{p} \right)^r
            \leq M^{r - 1} \cdot \left( \frac{\left( d_2 - 1 \right)^2}{p} \right)^r
        \end{align*}
        where the last two inequalities follow from \Cref{Equ: Arion GTDS linear probability} and the assumption that only $M$ approximations can be considered per round. \qed
    \end{proof}

    Provided that $M < \frac{p}{\left( d_2 - 1 \right)^2}$, then the inequality is always less than $1$.
    For the restricted differential hull we can then estimate the security level of \Arion via
    \begin{align}
        &M^{r - 1} \cdot \left( \frac{\left( d_2 - 1 \right)^2}{p} \right)^r \leq 2^{-\kappa} \\
        &\Longrightarrow \kappa \leq r \cdot \Big( \log_2 \left( p \right) - 2 \cdot \log_2 \left( d_2 - 1 \right) \Big) - \left( r - 1 \right) \cdot \log_2 \left( M \right).
    \end{align}
    In \Cref{Tab: security level linear cryptanalysis restricted hull} we report round number for various field sizes that achieve at least $128$ bit security for a restricted hull of size $\sqrt{p}$ at round level.

    \begin{table}[H]
        \centering
        \caption{Security level of \Arion against linear cryptanalysis for $p \geq 2^N$ and $d_2 \leq 257$ with a restricted linear hull of size $2^{N / 2}$ at round level.}
        \label{Tab: security level linear cryptanalysis restricted hull}
        \begin{tabular}{ M{0.8cm} | M{0.8cm} | M{1.5cm} }
            \toprule
            $r$ & $N$   & $\kappa$ (bits) \\
            \midrule

            $5$ & $60$  & $139$ \\
            $4$ & $120$ & $267$ \\
            $3$ & $250$ & $475$ \\

            \bottomrule
        \end{tabular}
    \end{table}

    \subsection{Other Statistical Attacks}
    By the definition of $\circulant \left( 1, \dots, n \right)$ a difference in one component can affect the whole state by a single round function call.
    Therefore, impossible differentials \cite{EC:BihBirSha99} and zero-correlation \cite{FSE:BogWan12,Bogdanov-Linear} can hardly be mounted on $3$ or more rounds.

    Boomerang attacks \cite{FSE:Wagner99,C:JouPey07} search for quartets that satisfy two differential paths simultaneously.
    For our target primes $\text{BLS}12$ and $\text{BN}254$ no differentials with high probability exist for more than two rounds, see \Cref{Equ: differential trail}.
    Therefore, a boomerang attack can hardly be mounted on $4$ rounds of \Arion \& \ArionHash.

    \section{Algebraic Attacks on \Arion}
    \subsection{Higher-Order Differential \& Interpolation Attacks}\label{Sec: interpolation attacks}
    Interpolation attacks \cite{FSE:JakKnu97} construct the polynomial vector representing a cipher without knowledge of the secret key.
    If such an attack is successful against a cipher, then an adversary can encrypt any plaintext without knowledge of the secret key.
    For a hash function the interpolated polynomial vector can be exploited to set up collision or forgery attacks.
    The cost of interpolating a polynomial depends on the number of monomials present in the polynomial vector representing the cipher function.
    Recall that any function $F: \Fqn \to \Fq$ can be represented by a unique polynomial $f \in \Fq [\mathbf{X}_n] = \Fq [x_1, \dots, x_n] / \left( x_1^q - x_1, \dots, x_n^q - x_n \right)$, thus at most $q^n$ monomials can be present in $f$.
    Clearly, if $f$ is dense, then an interpolation attack cannot be done faster than exhaustive search.

    Let $\mathcal{R}^{(i)} \subset \Fp [\mathbf{X}_n]$ denote the polynomial vector of the \Arionpi round function, we say that $\mathcal{R}^{(i + 1)} \circ \mathcal{R}^{(i)}$ has a degree overflow if we have to reduce with at least one of the field equations to compute the unique representation in $\Fp [\mathbf{X}_n]$.
    As we saw in \Cref{Tab: degree overflow}, for the fields $\text{BLS}12$ and $\text{BN}254$ some of our specified \Arion parameters already achieve a degree overflow in the first round.
    Moreover, after the first round we expect that terms
    \begin{equation}
        \left( \sum_{i = 1}^{n} i \cdot x_i \right)^e + \sum_{i = 1}^{n} i \cdot x_i
    \end{equation}
    are present in every branch.
    After another application of the round function we expect to produce the terms
    \begin{equation}\label{Equ: high degree terms}
        \left( \left( \sum_{i = 1}^{n} i \cdot x_i \right)^e + \sum_{i = 1}^{n} i \cdot x_i \right)^e \mod \left( x_1^p - x_1, \dots, x_n^p - x_n \right)
    \end{equation}
    in every branch.
    By our specification $e$ is the inverse exponent of a relatively low degree permutation, therefore we expect that the degrees of some of the aforementioned terms are close to $p - 2$.
    In particular, this is the case if $e^2 \geq p$.
    In addition, we also expect that a big fraction of the monomials in $\Fp [\mathbf{X}_n]$ is present in the components of the polynomial vector of \Arionpi after at least two iterations.
    Although an adversary can nullify the circulant matrix that is applied to the input vector in some attack scenarios, he cannot do so for later rounds.
    Therefore, we expect that at least after three rounds terms similar to \Cref{Equ: high degree terms} are present in the polynomial vector of \Arionpi.
    Further, to frustrate Meet-in-the-Middle (MitM) attacks we require that the number of rounds $r \geq 4$.

    We implemented \Arion in \verb!SageMath! \cite{SageMath} to compute the density of \Arionpi for small primes.
    Our findings in \Cref{Tab: density experiment} suggest that after two rounds \Arionpi has already almost full density over $\Fp$.
    \begin{table}[H]
        \centering
        \caption{Observed minimum density for \Arionpi for small primes, $n = 3, 4, 5$ and $d_1, d_2 = 3, 5$.}
        \label{Tab: density experiment}
        \begin{tabular}{ M{5mm} | M{5mm} | M{2.8cm} | M{2.3cm} | M{2.8cm} }
            \toprule
            $p$ & $r$ & Minimum density after $2$ rounds & Degree after $2$ rounds & Univariate degree after $2$ rounds \\
            \midrule

            $11$ & $6$ & $\geq 82 \%$ & $n \cdot (p - 1) - 1$ & $p - 1$ \\
            $13$ & $6$ & $\geq 91 \%$ & $n \cdot (p - 1) - 1$ & $p - 1$ \\
            $17$ & $6$ & $\geq 91 \%$ & $n \cdot (p - 1) - 1$ & $p - 1$ \\
            $19$ & $6$ & $\geq 92 \%$ & $n \cdot (p - 1) - 1$ & $p - 1$ \\
            $23$ & $6$ & $\geq 90 \%$ & $n \cdot (p - 1) - 1$ & $p - 1$ \\

            \bottomrule
        \end{tabular}
    \end{table}

    If after two rounds the density of \Arionpi is $\geq 0.8 \cdot p^n$, then for $p = 2^{60}$ and $n = 3$ already more than $2^{128}$ terms are present in \Arionpi.
    Therefore, we expect that \Arion resists interpolation attacks with the full key.
    An adversary can improve the capabilities of an interpolation attack by guessing parts of the key.
    If he can guess $n - 1$ parts correctly, then we expect that all univariate polynomials in \Arion have degree close to $p$.
    To retrieve the remaining key an adversary has to factor at least one of the polynomials, due to the high degree his best choice to perform a greatest common divisor computation.
    We can then estimate the complexity with $\mathcal{O} \big( p \cdot \log (p) \big)$.
    Therefore, for our target primes $\text{BLS}12$ and $\text{BN}254$ this complexity always exceeds the $128$ bit security target after two rounds.

    For an interpolation attack on \ArionHash we have a similar scenario as if we guessed parts of the key of \Arion.
    At worst only one input of \ArionHash is unknown to the adversary, he then still has to factor a polynomial of degree close to $p$, as mentioned before we do not expect that this defeats the $128$ bit security target after two rounds.

    Higher-order differential attacks \cite{FSE:Knudsen94,Lai-HigherOrder,C:BCDELLNPSTW20} exploit that higher differentials will vanish at some point.
    Since the density of \Arionpi exceeds $\geq 0.8 \cdot p^n$ after two rounds and its univariate degrees are close to $p$, we do not expect that higher-order differentials and distinguishers on \Arion \& \ArionHash can undermine the $128$ bit security target for our target primes $\text{BLS}12$ and $\text{BN}254$ after two rounds.

    \subsection{Integral Attacks}\label{Sec: integral attacks}
    The notion of integral cryptanalysis was introduced by Knudsen \& Wagner \cite{FSE:KnuWag02} and generalized to arbitrary finite fields by Beyne et al.\ \cite{C:BCDELLNPSTW20}.
    It is based on the following properties of polynomial valued functions.
    \begin{prop}[{see \cite[Proposition~1, 2]{C:BCDELLNPSTW20}}]
        Let $\Fq$ be a finite field, and let $F: \Fqn \to \Fq$ be a polynomial valued function.
        \begin{enumerate}
            \item Let $V \subset \Fqn$ be an affine subspace of dimension $k$.
            If $\degree{F} < k \cdot \left( q - 1 \right)$, then
            \begin{equation*}
                \sum_{\mathbf{x} \in V} F (\mathbf{x}) = 0.
            \end{equation*}

            \item Let $G_1, \dots, G_n \subset \Fq^\times$ be multiplicative subgroups, and let $G = \bigtimes_{i = 1}^{n} G_i$ their product group.
            If $\deg_{x_i} \left( F \right) < \left| G_i \right|$ for all $1 \leq i \leq n$, then
            \begin{equation*}
                \sum_{\mathbf{x} \in G} F (\mathbf{x}) - \left| G \right| \cdot F(\mathbf{0}) = 0.
            \end{equation*}
        \end{enumerate}
    \end{prop}
    \begin{proof}
        (1) was proven in \cite[Corollary~1]{C:BCDELLNPSTW20}, for (2) we index the terms of $F$ by $j$ and then rearrange the sum
        \begin{align*}
            \sum_{\mathbf{x} \in G} F (\mathbf{x})
            &= \sum_{x_1 \in G_1, \dots, x_n \in G_n} \ \sum_{j \geq 0} a_j \cdot x_1^{k_{1, j}} \cdots x_n^{k_{n, j}} \\
            &= \sum_{j \geq 0} \ \sum_{x_1 \in G_1, \dots, x_n \in G_n} a_j \cdot x_1^{k_{1, j}} \cdots x_n^{k_{n, j}}.
        \end{align*}
        For the zero term this sums to $\left| G \right| \cdot F(\mathbf{0})$, now assume that at least one $k_{i, j} > 0$, say $k_{1, j}$, then we further rearrange
        \begin{equation*}
            \sum_{x_1 \in G_1, \dots, x_n \in G_n} a_j \cdot x_1^{k_{1, j}} \cdots x_n^{k_{n, j}} = \sum_{x_2 \in G_1, \dots, x_n \in G_n} \left( \sum_{x_1 \in G_1} a_j \cdot x_1^{k_{1, j}} \cdots x_n^{k_{n, j}} \right).
        \end{equation*}
        For the inner sum on the right-hand side we can consider $x_2, \dots, x_n$ as non-zero field elements, by assumption $k_{1, j} < \left| G_1 \right|$ so by \cite[Proposition 2]{C:BCDELLNPSTW20} the sum vanishes.
        This proves the claim. \qed
    \end{proof}

    Let us investigate the capabilities of integral distinguishers on \Arionpi.
    Without loss of generality we can ignore the application $\circulant (1, \dots, n)$ before the first \Arion GTDS.
    To minimize the degree after the first \Arion GTDS we fix all inputs except $x_1$, then only the first component is non-constant in $x_1$, but after application of the first affine layer all components contain the monomial $x_1^{d_1}$.
    With our findings from interpolation attacks we expect that two additional rounds of \Arionpi are sufficient to increase the degree in $x_1$ of all components to at least $p - 2$.

    Therefore, for our target fields $\text{BLS}12$ and $\text{BN}254$ we do not expect that integral attacks can invalidate our $128$ bit security target after two rounds of \Arion \& \ArionHash.
    On the other hand, if \Arion or \ArionHash are instantiated over prime fields $p < 2^{128}$, then integral attacks could be a non-negligible threat since the degree of a univariate function can never exceed $p$.

    \section{Gr\"obner Basis Analysis}\label{Sec: Groebner basis attacks}
    In a Gr\"obner basis attack \cite{Buchberger,Cox-Ideals} the adversary represents a cryptographic function as fully determined system of polynomial equations and then solves for the solutions of the system.
    Since the system is fully determined at least one solution of the polynomial system must contain the quantity of interest, e.g.\ the key of a block cipher or the preimage of a hash function.
    In general, a Gr\"obner basis attack proceeds in four steps:
    \begin{enumerate}
        \item Model the cryptographic function with a (iterated) system of polynomials.

        \item Compute a Gr\"obner basis with respect to an efficient term order, e.g., the degree reverse lexicographic (DRL) order.

        \item Perform a term order conversion to an elimination order, e.g., the lexicographic (LEX) order.

        \item Solve the univariate equation.
    \end{enumerate}
    Let us for the moment assume that an adversary has already found a Gr\"obner basis and discuss the complexity of the remaining steps.
    Let $k$ be a field, let $I \subset k [x_1, \dots, x_n]$ be a zero-dimensional ideal modeling a cryptographic function, and let $d = \dim_k \left( k [x_1, \dots, x_n] / I \right)$ be the $k$-vector space dimension of the quotient space.
    With the original FGLM algorithm \cite{Faugere-FGLM} the complexity of term order conversion is $\mathcal{O} \left( n \cdot d^3 \right)$, but improved versions with probabilistic methods achieve $\mathcal{O} \left( n \cdot d^\omega \right)$ \cite{Faugere-SubCubic}, where $2 \leq \omega < 2.3727$, and sparse linear algebra algorithms \cite{Faugere-SparseFGLM} achieve $\mathcal{O} \left( \sqrt{n} \cdot d^{2 + \frac{n - 1}{n}} \right)$.
    To find the roots of the univariate polynomial from the elimination Gr\"obner basis\footnote{
        If $I$ is a radical ideal, then the degree of the univariate polynomial in the elimination Gr\"obner is indeed $d$, though for non-radical ideals the degree can be larger than $d$.
        For a simple example in the non-radical case consider $(x^2) \subset k [x]$.
    } $f \in \Fq [x]$ with $d = \degree{f}$ most efficiently we perform a greatest common divisor method that has recently been described in \cite[§3.1]{ToSC:BBLP22}.
    \begin{enumerate}
        \item Compute $g = x^q - x \mod f$.

        The computation of $x^q \mod f$ requires $\mathcal{O} \Big( d \cdot \log \left( q \right) \cdot \log \left( d \right) \cdot \log \big( \log \left( d \right) \big) \Big)$ field operations with a double-and-add algorithm.

        \item Compute $h = \gcd \left( f , g \right)$.

        By construction $h$ has the same roots as $f$ in $\Fq$ since $h = \gcd \left( f, x^q - x \right)$, but its degree is likely to be much lower.

        This step requires $\mathcal{O} \left( d \cdot \log \left( d \right)^2 \cdot \log \big( \log \left( d \right) \big) \right)$ field operations.

        \item Factor $h$.

        In general, the polynomial $f$ coming from a $0$-dimensional Gr\"obner basis has only a few roots in $\Fq$.

        Thus, this step is negligible in complexity.
    \end{enumerate}
    Note that this method is only performative if $\degree{f} < q$ else one has to exchange the roles of $f$ and the field equation.
    Overall solving the polynomial system with probabilistic methods requires
    \begin{equation}\label{Equ: complexity probabilistic solving}
        \mathcal{O} \left( n \cdot d^\omega + d \cdot \log \left( q \right) \cdot \log \left( d \right) \cdot \log \big( \log \left( d \right) \big) + d \cdot \log \left( d \right)^2 \cdot \log \big( \log \left( d \right) \big) \right)
    \end{equation}
    field operations, and with deterministic methods
    \begin{equation}\label{Equ: complexity deterministic solving}
        \mathcal{O} \left( \sqrt{n} \cdot d^{2 + \frac{n - 1}{n}} + d \cdot \log \left( q \right) \cdot \log \left( d \right) \cdot \log \big( \log \left( d \right) \big) + d \cdot \log \left( d \right)^2 \cdot \log \big( \log \left( d \right) \big) \right)
    \end{equation}
    field operations.

    We must stress that we are unable to quantify the success probability of the probabilistic term order conversion.
    The probabilistic analysis of \cite[\S 5.1]{Faugere-SparseFGLM} requires that the ideal is radical and that the homogenization of the DRL Gr\"obner basis is a regular sequence.
    We neither have a proof nor a disproof of these technical requirements for the \Arion and \ArionHash polynomial systems.
    Therefore, we cannot quantify how capable the probabilistic approach is aside from its complexity estimate.

    Let us now discuss the complexity of Gr\"obner basis computations.
    Today, the most efficient algorithms to compute Gr\"obner bases are Faug\`{e}re's linear algebra-based algorithms F4 \cite{Faugere-F4} and F5 \cite{Faugere-F5}.
    The main idea of linear algebra-based Gr\"obner basis algorithms can be traced back to Lazard \cite{Lazard-Groebner}.
    Let $\mathcal{F} = \{ f_1, \dots, f_m \} \subset P = k [x_1, \dots, x_n]$ be a finite set of homogeneous polynomials over a field $k$.
    The \textit{homogeneous Macaulay matrix} $M_d$ of degree $d$ has columns indexed by monomials in $s \in P_d$ and rows indexed by polynomials $t \cdot f_j$, where $t \in P$ is a monomial such that $\degree{t \cdot f_j} = d$.
    The entry of row $t \cdot f_j$ at column $s$ is then the coefficient of the monomial $s$ in $t \cdot f_j$.
    If $\mathcal{F}$ is an inhomogeneous system of polynomials, then one replaces $M_d$ by $M_{\leq d}$ and the degree equality by an inequality.
    If one fixes a term order $>$ on $P$, then by performing Gaussian elimination on $M_d$ respectively $M_{\leq d}$ for a large enough value of $d$ one produces a $>$-Gr\"obner basis of $\mathcal{F}$.
    The least such $d$ is called the solving degree $\solvdeg_> \left( \mathcal{F} \right)$ of $\mathcal{F}$.
    (The notion of solving degree was first introduced in \cite{Ding-SolvingDegree}, though we use the definition of \cite{Caminata-SolvingPolySystems}.)
    The main improvement of F4/5 over Lazard's method is the choice of efficient selection criteria.
    Conceptually, the Macaulay matrix will contain many redundant rows, if one is able to avoid many of these rows with selection criteria, then the running time of an algorithm will improve.
    Nevertheless, with the notion of the solving degree it is possible to upper bound the maximal size of the computational universe of F4/5.
    It is well-known that the number of monomials in $P$ of degree $d$ is given by the binomial coefficient
    \begin{equation}
        N (n, d) = \binom{n + d - 1}{d}.
    \end{equation}
    We can now upper bound the size of the Macaulay matrix by $m \cdot d \cdot N (n, d) \times d \cdot N (n, d)$.
    Overall, we can bound the complexity of Gaussian elimination on the Macaulay matrix $M_{\leq d}$ by
    \begin{equation}\label{Equ: complexity estimate}
        \mathcal{O} \Bigg( \binom{n + d - 1}{d}^{\omega} \Bigg),
    \end{equation}
    where we absorb $m$ and $d$ in the implied constant since in general $N (n, d) \gg m, d$ and $\omega \geq 2$ is a linear algebra constant.

    In the cryptographic literature a generic approach to bound the solving degree is the so-called \textit{Macaulay bound}.
    Assume that $\degree{f_1} \geq \ldots \geq \degree{f_m}$ and let $l = \min \{ m, n + 1 \}$, then the Macaulay bound of $\mathcal{F}$ is given by
    \begin{equation}\label{Equ: Macaulay bound}
        \mb_\mathcal{F} = \sum_{i = 1}^{l} \degree{f_1} + \ldots + \degree{f_l} - l + 1.
    \end{equation}
    Up to date there are two known cases when one indeed has that $\solvdeg_{DRL} \left( \mathcal{F} \right) \leq \mb_\mathcal{F}$:
    \begin{enumerate}
        \item $\mathcal{F}$ is regular \cite{Bardet-AsymptoticIndex,Bardet-AsymptoticDegree}, or

        \item $\mathcal{F}$ is in generic coordinates \cite[Theorem~9, 10]{Caminata-SolvingPolySystems}.
    \end{enumerate}
    We stress that we were unable to prove that one of the two cases applies to \Arion or \ArionHash.
    Therefore, we can only hypothesize that the Macaulay bound upper bounds the solving degree.

    From a designer's perspective, during all our small scale experiments the vector space dimension of the quotient space behaved more stable with respect to the chosen primes, branch sizes and round numbers than the observed solving degree.
    \textbf{Thus, all our extrapolated security claims of \Arion and \ArionHash with respect to Gr\"obner basis attacks are expressed in the complexity of solving for the solutions of the polynomial system after a Gr\"obner basis has been found.}

    Moreover, in all our experiments we observed that the quotient space dimension grows exponentially in $r$ and the base only depends on $n, d_1$ and $d_2$.
    We hypothesize that this behaviour is invariant under the chosen primes and parameters.
    Therefore, we will do an ``educated guess'' from our small scale experiments to obtain the base and extrapolate our findings for our security analysis.

    \subsection{\Arion}
    For the security of \Arion against Gr\"obner basis attacks we consider the following polynomial model:
    \begin{enumerate}[label=(\roman*)]
        \item We do not consider a key schedule, i.e., in every round we add the same key $\mathbf{k} = \left( k_1, \dots, k_n \right) \in \Fpn$.

        \item We use a single plain/cipher pair $\mathbf{p}, \mathbf{c} \in \Fpn$ given by \Arion to set up a fully determined polynomial system $\mathcal{F}$.

        \item\label{Item: naive model} For $1 \leq i \leq r - 1$ we denote with $\mathbf{x}^{(i)} = \left( x_1^{(i)}, \dots, x_n^{(i)} \right)$ the intermediate state variables, in addition we set $\mathbf{x}^{(0)} = \circulant \left( 1, \dots, n \right) \cdot \left( \mathbf{p} + \mathbf{k} \right)$ and $\mathbf{x}^{(r)} = \mathbf{c}$.
        Further, for $1 \leq i \leq r$ we denote with $z^{(i)}$ an auxiliary variable.
        For our polynomial model we choose a slight modification $\tilde{\mathcal{F}} = \{ \tilde{f}_1, \dots, \tilde{f}_n \}$ of the \Arion GTDS $\mathcal{F} = \{ f_1, \dots, f_n \}$, see \Cref{Def: GTDS}, where we set $\tilde{f}_n (x_n) = x_n$.
        For $1 \leq i \leq n - 1$ the $\tilde{f}_i$ follow the same iterative definition as the original polynomials in the GTDS, and we modify $\tilde{\sigma}_{i + 1, n} = x_n + z_n + \sum_{j = i + 1}^{n - 1} x_j + \tilde{f}_j$.
        We consider the following polynomial model as the naive model $\mathcal{F}_\text{naive}$ for \Arion
        \begin{equation*}
            \begin{split}
                \circulant \left( 1, \dots, n \right)
                \tilde{\mathcal{F}} \left( \mathbf{\hat{x}}^{(i - 1)} \right)
                + \mathbf{c}_i + \mathbf{k} -
                \mathbf{x}^{(r)}
                &= \mathbf{0}, \\
                \left( x_n^{(i - 1)} \right)^e - z^{(i)} &= 0,
            \end{split}
        \end{equation*}
        where $\mathbf{\hat{x}}^{(i)} = \left( x_1^{(i)}, \dots, x_{n - 1}^{(i)}, z^{(i)} \right)$.

        \item\label{Item: efficient model} Obviously the naive polynomial system $\mathcal{F}_\text{naive}$ contains high degree equations given by the power permutation $x^e$.
        Though, if we replace the auxiliary equations by
        \begin{equation*}
            x_n^{(i)} - \left( z^{(i)} \right)^{d_2} = 0,
        \end{equation*}
        then we obtain a polynomial system $\mathcal{F}_\Arion$ whose polynomials are of small degree.
        Further, we expect the arithmetic of $\mathcal{F}_\Arion$ to be independent of the chosen prime, i.e., for primes $p, q \in \mathbb{P}$ such that $\gcd \left( d_i, p - 1 \right) = 1 = \gcd \left( d_i, q - 1 \right)$ we expect no notable difference for the complexity of a Gr\"obner basis attack.
    \end{enumerate}

    \begin{lem}
        Let $\Fp$ be a finite field, let $\mathcal{F}_\text{naive}$ and $\mathcal{F}_\Arion$ be the polynomial models from \ref{Item: naive model} and \ref{Item: efficient model}, and let $F$ be the ideal of all field equations in the polynomial ring of $\mathcal{F}_\text{naive}$ and $\mathcal{F}_\Arion$.
        Then
        \begin{equation*}
            \left( \mathcal{F}_\text{naive} \right) + F = \left( \mathcal{F}_\Arion \right) + F.
        \end{equation*}
    \end{lem}
    \begin{proof}
        By definition, we have that $\left( x_n^{(i)} \right)^e \equiv z^{(i)} \mod \left( \mathcal{F}_\text{naive} \right) + F$, by raising this congruence to the $d_2$\textsuperscript{th} power yields
        \begin{equation*}
            \left( \left( x_n^{(i)} \right)^e \right)^{d_2} \equiv x_n^{(i)} \equiv \left( z^{(i)} \right)^{d_2} \mod \left( \mathcal{F}_\text{naive} \right) + F
        \end{equation*}
        which proves the claim. \qed
    \end{proof}

    I.e., on the solutions that solely come from the finite field $\Fp$, which are the solutions of cryptographic interest, the varieties corresponding to $\mathcal{F}_\text{naive}$ and $\mathcal{F}_\Arion$ coincide, so
    \begin{equation}
        \mathcal{V} \left( \mathcal{F}_\text{naive} \right) \cap \Fpn = \mathcal{V} \left( \mathcal{F}_\Arion \right) \cap \Fpn.
    \end{equation}
    Thus, $\mathcal{F}_\Arion$ is indeed a well-founded model for \Arion.

    To compute the Macaulay bound of $\mathcal{F}_\Arion$ we first apply $\circulant (1, \dots, n)^{-1}$ in every round to cancel the mixing of the components, then we use \Cref{Lem: degrees in GTDS} with $e = 1$ to compute the degree in the $i$\textsuperscript{th} component.
    Further, we have to account for the auxiliary equation, overall we yield the Macaulay bound
    \begin{equation}\label{Equ: cipher Macaulay bound}
        \begin{split}
            \mb_{n, d_1, d_2} (r)
            &= r \cdot \left( d_2 + 1 + \sum_{i = 1}^{n - 1} \left( 2^{n - i} \cdot (d_1 + 1) - d_1 \right) \right) + 1 - r \cdot \left( n + 1 \right) \\
            &= r \cdot \big( d_2 + 2 \cdot(d_1 + 1) \cdot \left( 2^{n - 1} - 1 \right) - (n - 1) \cdot d_1 - n \big) + 1
        \end{split}
    \end{equation}
    We stress that we were unable to prove or disprove that $\mathcal{F}_\Arion$ is regular or in generic coordinates, thus we can only hypothesize the Macaulay bound as measure for the complexity of linear algebra-based Gr\"obner basis algorithms.

    We implemented $\mathcal{F}_\Arion$ in the \verb!OSCAR! computer algebra system \cite{OSCAR} and computed the Gr\"obner basis with its F4 implementation.
    Unfortunately, the log function of F4 only prints the current working degree of the algorithm not its current solving degree.
    As remedy, we estimate the empirical solving degree as follows: We sum up all positive differences of the working degrees between consecutive steps and add this quantity to the largest input polynomial degree.
    After computing the Gr\"obner basis we also computed the $\Fp$-vector space dimension of the quotient ring.
    We conducted our experiments with the primes
    \begin{equation}\label{Equ: primes for experiment}
        \underbrace{
            \begin{aligned}
                p_1 &= 1013, \\
                p_2 &= 10007,
            \end{aligned}
        }_{\gcd \left( 3, p_i - 1 \right) = 1} \qquad\text{and}\qquad
        \underbrace{
            \begin{aligned}
                p_3 &= 1033, \\
                p_4 &= 15013.
            \end{aligned}
        }_{\gcd \left( 5, p_i - 1 \right) = 1}
    \end{equation}
    All computations were performed on an AMD EPYC-Rome (48) CPU with 94 GB RAM.

    In \Cref{Fig: cipher n=2} we record our empirical results for $n = 2$ and in \Cref{Fig: cipher n=3} we record our empirical results for $n = 3$.
    From experiments, we hypothesize that for $r > 1$ the solving degree is indeed bounded by the Macaulay bound and that the quotient space dimension grows exponential in $r$.
    \begin{figure}[H]
        \centering
            \begin{tikzpicture}[declare function={Sd(\n,\d,\e,\r) = \r * (\e + 2 * (\d + 1) * (2^(\n - 1) - 1) - (\n - 1) - \n) + 1;Dim3(\r) = 35^\r;Dim5(\r) = 49^\r;}]
        \pgfplotsset{every axis/.append style={width=0.44\linewidth,title style={align=center}}}
        \begin{axis}[name=axis1,
                     xlabel=$r$,
                     ylabel=$\solvdeg_{DRL} \left( \mathcal{F}_\Arion \right)$,
                     legend entries={$d_1 = 3$, $d_1 = 5$, $\mb_{2, 3, 7} (r)$, $\mb_{2, 5, 7} (r)$},
                     legend style={at={(0.02, 0.98)}, anchor=north west, nodes={scale=0.4, transform shape}},
                     grid=both,
                     grid style={line width=.1pt, draw=black!50}, major grid style={line width=.2pt,draw=black!100},
                     ymin=0,
                     xtick=data,
                     ytick={0, 20, ..., 1000},
                     minor y tick num=3]
            \addplot [color=blue,
                      mark=x] coordinates {
                (1, 14)
                (2, 24)
                (3, 24)
            };
            \addplot [color=cyan,
                      mark=x] coordinates {
                (1, 15)
                (2, 26)
                (3, 26)
            };
            \addplot [domain=1:3,
                      color=red] {
                Sd(2, 3, 7, x)
            };
            \addplot [domain=1:3,
                      color=purple] {
                Sd(2, 5, 7, x)
            };
        \end{axis}
        \begin{axis}[at={(axis1.outer north east)},
                     anchor=outer north west,
                     name=axis2,
                     xlabel=$r$,
                     ylabel=$\dim_{\mathbb{F}_p} \left( \mathcal{F}_\Arion \right)$,
                     legend entries={$d_1 = 3$, $d_1 = 5$, $35^r$, $49^r$},
                     legend style={at={(0.02, 0.98)}, anchor=north west, nodes={scale=0.4, transform shape}},
                     grid=both,
                     grid style={line width=.1pt, draw=black!50}, major grid style={line width=.2pt,draw=black!100},
                     ymode=log,
                     log basis y=10,
                     xtick=data]
            \addplot [color=blue,
                      mark=x] coordinates {
                (1, 35)
                (2, 1225)
                (3, 42875)
            };
            \addplot [color=cyan,
                      mark=x] coordinates {
                (1, 49)
                (2, 2401)
                (3, 117649)
            };
            \addplot [domain=1:3,
                      color=red] {
                Dim3(x)
            };

            \addplot [domain=1:3,
                      color=purple] {
                Dim5(x)
            };
        \end{axis}
    \end{tikzpicture}
        \caption{Experimental solving degree and vector space dimension of the quotient ring for \Arion with $n = 2$ and $d_2 = 7$.}
        \label{Fig: cipher n=2}
    \end{figure}
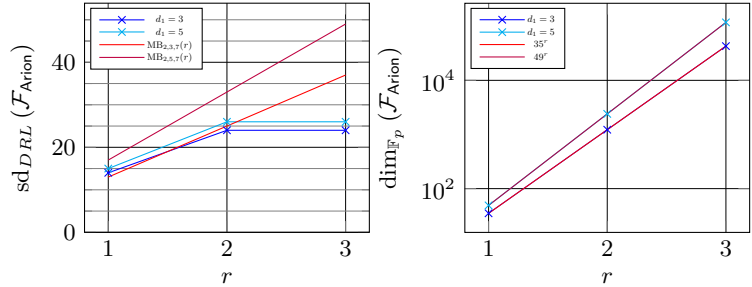
    \clearpage 
    \begin{figure}[H]
        \centering
            \begin{tikzpicture}[declare function={Sd(\n,\d,\e,\r) = \r * (\e + 2 * (\d + 1) * (2^(\n - 1) - 1) - (\n - 1) - \n) + 1;Dim3(\r) = 175^\r;Dim5(\r)=343^\r;}]
        \pgfplotsset{every axis/.append style={width=0.4\linewidth,title style={align=center}}}
        \begin{axis}[name=axis1,
                     xlabel=$r$,
                     ylabel=$\solvdeg_{DRL} \left( \mathcal{F}_\Arion \right)$,
                     legend entries={$d_1 = 3$, $d_1 = 5$, $\mb_{3, 3, 7} (r)$, $\mb_{3, 5, 7} (r)$},
                     legend style={at={(0.02, 0.98)}, anchor=north west, nodes={scale=0.4, transform shape}},
                     grid=both,
                     grid style={line width=.1pt, draw=black!50}, major grid style={line width=.2pt,draw=black!100},
                     ymin=0,
                     xtick=data,
                     ytick={0, 40, ..., 1000},
                     minor y tick num=3]
            \addplot [color=blue,
                      mark=x] coordinates {
                (1, 31)
                (2, 45)
            };
            \addplot [color=cyan,
                      mark=x] coordinates {
                (1, 43)
                (2, 61)
            };
            \addplot [domain=1:2,
                      color=red] {
                Sd(3, 3, 7, x)
            };
            \addplot [domain=1:2,
                      color=purple] {
                Sd(3, 5, 7, x)
            };
        \end{axis}
        \begin{axis}[at={(axis1.outer north east)},
                     anchor=outer north west,
                     name=axis2,
                     xlabel=$r$,
                     ylabel=$\dim_{\Fp} \left( \mathcal{F}_\Arion \right)$,
                     legend entries={$d_1 = 3$, $d_1 = 5$, $175^r$},
                     legend style={at={(0.02, 0.98)}, anchor=north west, nodes={scale=0.4, transform shape}},
                     grid=both,
                     grid style={line width=.1pt, draw=black!50}, major grid style={line width=.2pt,draw=black!100},
                     ymode=log,
                     log basis y=10,
                     xtick=data]
            \addplot [color=blue,
                      mark=x] coordinates {
                (1, 175)
                (2, 30625)
            };
            \addplot [color=cyan,
                      mark=x] coordinates {
                (1, 343)
                (2, 117649)
            };
            \addplot [domain=1:2,
                      color=red] {
                Dim3(x)
            };
            \addplot [domain=1:2,
                      color=purple] {
                Dim5(x)
            };
        \end{axis}
    \end{tikzpicture}
        \caption{Experimental solving degree and vector space dimension of the quotient ring for \Arion with $n = 3$ and $d_2 = 7$.}
        \label{Fig: cipher n=3}
    \end{figure}
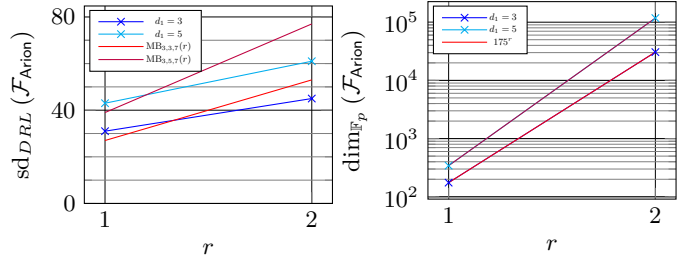

    To better understand the growth of the base of the quotient space dimension we computed the quotient space dimension for $n \leq 4$ and $r = 1$, see \Cref{Tab: cipher empirical quotient space dimension}.
    \begin{table}[H]
        \centering
        \caption{Empirical growth of the vector space dimension of the quotient space of \Arion with $r = 1$.}
        \label{Tab: cipher empirical quotient space dimension}
        \begin{tabular}{ M{5mm} | M{5mm} | M{7mm} | M{1.6cm} }
            \toprule
            $n$ & $d_1$ & $d_2$ & $\dim_{\Fp, \text{emp}}$ \\
            \midrule

            $2$ & $3$ & $7$   & $35$    \\
            $2$ & $3$ & $257$ & $1285$  \\
            $3$ & $3$ & $7$   & $175$   \\
            $4$ & $3$ & $7$   & $875$   \\

            $2$ & $5$ & $7$   & $49$    \\
            $2$ & $5$ & $257$ & $1799$  \\
            $3$ & $5$ & $7$   & $343$   \\

            \bottomrule
        \end{tabular}
    \end{table}

    From our experiments we conjecture that the quotient space dimension grows or is bounded via
    \begin{equation}\label{Equ: cipher quotient space dimension}
        \dim_{\Fp} \left( \mathcal{F}_\Arion \right) (n, r, d_1, d_2) = \left( d_2 \cdot \left( d_1 + 2 \right)^{n - 1} \right)^r.
    \end{equation}

    We use two approaches to estimate the cost of Gr\"obner basis computations for $\mathcal{F}_\Arion$.
    First, note that if we exclude one polynomial from $\mathcal{F}_\Arion$, then we do not have a zero-dimensional polynomial system anymore.
    Therefore, the highest non-trivial lower bound for the solving degree of $\mathcal{F}_\Arion$ has to be the highest degree in the \Arion polynomial system, i.e.,
    \begin{equation}\label{Equ: minimal solving degree}
        \max \left\{ d_2, 2^{n - 1} \cdot (d_1 + 1) - d_1 \right\} \leq \solvdeg_{DRL} \left( \mathcal{F}_\Arion \right).
    \end{equation}
    We combine this lower bound with \Cref{Equ: complexity estimate} to derive the ``minimal possible`` Gr\"obner basis complexity estimate.
    Second, under the assumption that the Macaulay bound always upper bounds the solving degree of \Arion we can combine \Cref{Equ: complexity estimate,Equ: Macaulay bound} to derive the ``maximal possible'' Gr\"obner basis complexity estimate.
    For ease of computation we approximated the binomial coefficient with
    \begin{equation}
        \binom{n}{k} \approx \sqrt{\frac{n}{\pi \cdot k \cdot (n - k)}} \cdot 2^{n \cdot H_2 (k / n)},
    \end{equation}
    where $H_2 (p) = -p \cdot \log_2 \left( p \right) - \left( 1 - p \right) \cdot \log_2 \left( 1 - p \right)$ denotes the binary entropy (cf.\ \cite[Lemma~17.5.1]{Cover-InformationTheory}).
    To estimate the cost of system solving we plug \Cref{Equ: cipher quotient space dimension} into \Cref{Equ: complexity probabilistic solving,Equ: complexity deterministic solving}.
    For all estimates we assume that our adversary has an optimal Gaussian elimination algorithm with $\omega = 2$.

    \textbf{We base the security of \Arion against Gr\"obner basis attacks solely on the complexity of solving the polynomial system.}
    I.e., even if an adversary can compute a Gr\"obner basis in $\mathcal{O} (1)$ it must be computationally infeasible to find a solution of the polynomial system.
    Thus, we estimate the security level $\kappa$ of \Arion against a Gr\"obner basis attack via
    \begin{equation}
        \kappa \leq \log_2 \big( \mathcal{O} \left( \text{System solving} \right) \big).
    \end{equation}
    \begin{table}[H]
        \centering
        \caption{Empirical cost estimation of Gr\"obner basis attacks on \Arion for primes $p \geq 2^{60}$ and $d_2 \geq 121$.
            The column GB min contains the complexity of Gr\"obner basis computations estimated with the highest polynomial degree in the \Arion polynomial system.
            The column GB MB contains the complexity Gr\"obner basis computations estimated via the Macaulay bound.
            We assume that the adversary has an optimal Gaussian elimination algorithm with $\omega = 2$.}
        \label{Tab: cipher Groebner basis complexity}
        \resizebox{0.7\textwidth}{!}{
            \begin{tabular}{ M{5mm} | M{5mm} | M{15mm} | M{15mm} | M{25mm} | M{25mm} }
                \toprule
                $n$ & $r$ & GB min (bits) & GB MB (bits) & Deterministic Solving (bits) & Probabilistic Solving (bits) \\
                \midrule

                \multicolumn{6}{ c }{$d_1 = 3$} \\
                \midrule

                $3$ & $4$ & $123$ & $129$ & $137$ & $96$  \\
                $3$ & $6$ & $162$ & $170$ & $207$ & $143$ \\

                $4$ & $4$ & $143$ & $160$ & $165$ & $115$ \\
                $4$ & $5$ & $166$ & $186$ & $207$ & $143$ \\

                $5$ & $3$ & $134$ & $164$ & $145$ & $101$ \\
                $5$ & $4$ & $162$ & $201$ & $194$ & $134$ \\

                $6$ & $3$ & $150$ & $208$ & $166$ & $115$ \\
                $6$ & $4$ & $181$ & $257$ & $222$ & $153$ \\

                $8$ & $2$ & $204$ & $244$ & $138$ & $96$  \\
                $8$ & $3$ & $279$ & $338$ & $208$ & $143$ \\

                \midrule
                \multicolumn{6}{c}{$d_1 = 5$} \\
                \midrule

                $3$ & $4$ & $123$ & $131$ & $149$ & $104$ \\
                $3$ & $5$ & $143$ & $153$ & $187$ & $129$ \\

                $4$ & $3$ & $118$ & $135$ & $136$ & $95$  \\
                $4$ & $5$ & $166$ & $195$ & $229$ & $158$ \\

                $5$ & $3$ & $134$ & $174$ & $162$ & $113$ \\
                $5$ & $4$ & $162$ & $215$ & $217$ & $149$ \\

                $6$ & $3$ & $172$ & $225$ & $187$ & $130$ \\

                $8$ & $2$ & $225$ & $263$ & $158$ & $110$ \\
                $8$ & $3$ & $309$ & $366$ & $238$ & $164$ \\

                \bottomrule
            \end{tabular}
        }
    \end{table}
    \clearpage 

    \subsection{\ArionHash}
    In a preimage attack on \ArionHash we are given a given hash value $\alpha \in \Fp$ and we have to find $\mathbf{x} \in \Fp^{r'}$ such that $\ArionHash (\mathbf{x}) = \alpha$.
    In a second-preimage attack we assume that we are given a message that consists of two input blocks, i.e.\ $\mathbf{y} = \left( \mathbf{y}_1, \mathbf{y}_2 \right) \in \left( \Fp^{r'} \right)^2$.
    Now we again have to find $\mathbf{x} \in \Fp^{r'}$ such that $\ArionHash (\mathbf{y}) = \ArionHash (\mathbf{x})$.
    Consequently, both preimage attacks on \ArionHash reduce to the same equation
    \begin{equation}\label{Equ: preimage attack}
        \text{\Arion -} \pi
        \begin{pmatrix}
            \mathbf{x}_\text{in} \\
            \texttt{IV}
        \end{pmatrix}
        =
        \begin{pmatrix}
            \alpha \\
            \mathbf{x}_\text{out}
        \end{pmatrix}
        ,
    \end{equation}
    where $\alpha \in \Fp$ is the output of \ArionHash, $\texttt{IV} \in \Fp^{c}$ is the initial value, and $\mathbf{x}_\text{in} = \left( x_{\text{in}, 1}, \dots, x_{\text{in}, r'} \right)$ and $\mathbf{x}_\text{out} = \left( x_{\text{out}, 2}, \dots, x_{\text{out}, n} \right)$ are indeterminates.
    Analog to \Arion we construct an iterated polynomial system where each round is modeled with a polynomial vector.
    Note that \Cref{Equ: preimage attack} is not fully determined if $n \geq 3$ and $r' > 1$, in such a case we have $n - 1 + r' > n$ many variables for $\mathbf{x}_\text{in}$ and $\mathbf{x}_\text{out}$ and $(r - 1) \cdot n$ many intermediate state variables, but we only have $r \cdot n$ many equations.
    In order to obtain a fully determined system the adversary either has to guess some values for $\mathbf{x}_\text{in}$ or $\mathbf{x}_\text{out}$, but each guess has success probability $1 / p$, so we can neglect this approach, or he has to add additional equations.
    If an adversary is unable to exploit additional algebraic structures of \Arionpi (which are unknown to us), then he has just one generic choice: he has to add field equations until the system is fully determined.
    This approach introduces polynomials of very high degree, thus we do not expect this attack to be feasible below our $128$ bit Gr\"obner basis security claim.
    Therefore, in the analysis of this section we always choose $c = n - 1$ to obtain a fully determined system.
    Note that the Macaulay bound for \ArionHash is identical to the one for \Arion, see \Cref{Equ: cipher Macaulay bound}.

    We implemented \ArionHash in the \verb!OSCAR! \cite{OSCAR} computer algebra system and computed the Gr\"obner basis of $\mathcal{F}_\ArionHash$ with F4.
    As initial value we chose $\texttt{IV} = \mathbf{0}^c$.
    Further, we computed the $\Fp$-vector space dimension of the quotient space for \ArionHash.
    We used the same primes as for \Arion, see \Cref{Equ: primes for experiment}.
    In \Cref{Fig: hash n=2} we record our empirical results for $n = 2$ and in \Cref{Fig: hash n=3} we record our empirical results for $n = 3$.
    From experiments, we hypothesize that the solving degree is indeed bounded by the Macaulay bound and that the quotient space dimension grows exponential in $r$.
    \clearpage 
    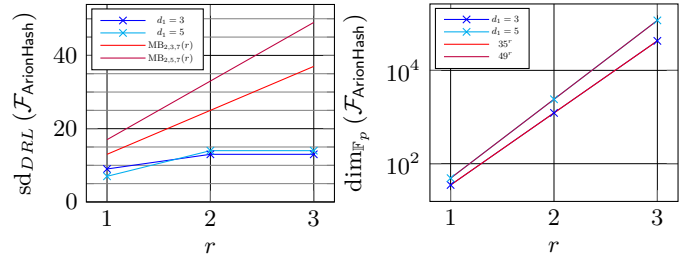
\begin{figure}[H]
        \centering
            \begin{tikzpicture}[declare function={Sd(\n,\d,\e,\r) = \r * (\e + 2 * (\d + 1) * (2^(\n - 1) - 1) - (\n - 1) - \n) + 1;Dim3(\r) = 35^\r;Dim5(\r) = 49^\r;}]
        \pgfplotsset{every axis/.append style={width=0.4\linewidth,title style={align=center}}}
        \begin{axis}[name=axis1,
                     xlabel=$r$,
                     ylabel=$\solvdeg_{DRL} \left( \mathcal{F}_\ArionHash \right)$,
                     legend entries={$d_1 = 3$, $d_1 = 5$, $\mb_{2, 3, 7} (r)$, $\mb_{2, 5, 7} (r)$},
                     legend style={at={(0.02, 0.98)}, anchor=north west, nodes={scale=0.4, transform shape}},
                     grid=both,
                     grid style={line width=.1pt, draw=black!50}, major grid style={line width=.2pt,draw=black!100},
                     ymin=0,
                     xtick=data,
                     ytick={0, 20, ..., 1000},
                     minor y tick num=3]
            \addplot [color=blue,
                      mark=x] coordinates {
                (1, 9)
                (2, 13)
                (3, 13)
            };
            \addplot [color=cyan,
                      mark=x] coordinates {
                (1, 7)
                (2, 14)
                (3, 14)
            };
            \addplot [domain=1:3,
                      color=red] {
                Sd(2, 3, 7, x)
            };
            \addplot [domain=1:3,
                      color=purple] {
                Sd(2, 5, 7, x)
            };
        \end{axis}
        \begin{axis}[at={(axis1.outer north east)},
                     anchor=outer north west,
                     name=axis2,
                     xlabel=$r$,
                     ylabel=$\dim_{\Fp} \left( \mathcal{F}_\ArionHash \right)$,
                     legend entries={$d_1 = 3$, $d_1 = 5$, $35^r$, $49^r$},
                     legend style={at={(0.02, 0.98)}, anchor=north west, nodes={scale=0.4, transform shape}},
                     grid=both,
                     grid style={line width=.1pt, draw=black!50}, major grid style={line width=.2pt,draw=black!100},
                     ymode=log,
                     log basis y=10,
                     xtick=data]
            \addplot [color=blue,
                      mark=x] coordinates {
                (1, 35)
                (2, 1225)
                (3, 42875)
            };
            \addplot [color=cyan,
                      mark=x] coordinates {
                (1, 49)
                (2, 2401)
                (3, 117649)
            };
            \addplot [domain=1:3,
                      color=red] {
                Dim3(x)
            };

            \addplot [domain=1:3,
                      color=purple] {
                Dim5(x)
            };
        \end{axis}
    \end{tikzpicture}
        \caption{Experimental solving degree and vector space dimension of the quotient ring for \ArionHash with $n = 2$ and $d_2 = 7$.}
        \label{Fig: hash n=2}
    \end{figure}
    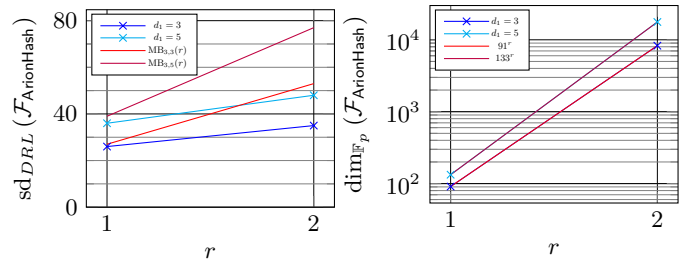
\begin{figure}[H]
        \centering
            \begin{tikzpicture}[declare function={Sd(\n,\d,\e,\r) = \r * (\e + 2 * (\d + 1) * (2^(\n - 1) - 1) - (\n - 1) - \n) + 1;Dim3(\r) = 91^\r;Dim5(\r) = 133^\r;}]
        \pgfplotsset{every axis/.append style={width=0.4\linewidth,title style={align=center}}}
        \begin{axis}[name=axis1,
                     xlabel=$r$,
                     ylabel=$\solvdeg_{DRL} \left( \mathcal{F}_\ArionHash \right)$,
                     legend entries={$d_1 = 3$, $d_1 = 5$, $\mb_{3, 3} (r)$, $\mb_{3, 5} (r)$},
                     legend style={at={(0.02, 0.98)}, anchor=north west, nodes={scale=0.4, transform shape}},
                     grid=both,
                     grid style={line width=.1pt, draw=black!50}, major grid style={line width=.2pt,draw=black!100},
                     ymin=0,
                     xtick=data,
                     ytick={0, 40, ..., 1000},
                     minor y tick num=3]
            \addplot [color=blue,
                      mark=x] coordinates {
                (1, 26)
                (2, 35)
            };
            \addplot [color=cyan,
                      mark=x] coordinates {
                (1, 36)
                (2, 48)
            };
            \addplot [domain=1:2,
                     color=red] {
                Sd(3, 3, 7, x)
            };
            \addplot [domain=1:2,
                      color=purple] {
                Sd(3, 5, 7, x)
            };
        \end{axis}
        \begin{axis}[at={(axis1.outer north east)},
                     anchor=outer north west,
                     name=axis2,
                     xlabel=$r$,
                     ylabel=$\dim_{\Fp} \left( \mathcal{F}_\ArionHash \right)$,
                     legend entries={$d_1 = 3$, $d_1 = 5$, $91^r$, $133^r$},
                     legend style={at={(0.02, 0.98)}, anchor=north west, nodes={scale=0.4, transform shape}},
                     grid=both,
                     grid style={line width=.1pt, draw=black!50}, major grid style={line width=.2pt,draw=black!100},
                     ymode=log,
                     log basis y=10,
                     xtick=data]
            \addplot [color=blue,
                      mark=x] coordinates {
                (1, 91)
                (2, 8281)
            };
            \addplot [color=cyan,
                      mark=x] coordinates {
                (1, 133)
                (2, 17689)
            };
            \addplot [domain=1:2,
                      color=red] {
                Dim3(x)
            };
            \addplot [domain=1:2,
                      color=purple] {
                Dim5(x)
            };
        \end{axis}
    \end{tikzpicture}
        \caption{Experimental solving degree and vector space dimension of the quotient ring for \ArionHash with $n = 3$ and $d_2 = 7$.}
        \label{Fig: hash n=3}
    \end{figure}

    To better understand the growth of the base of the quotient space dimension we computed the quotient space dimension for $n \leq 5$ and $r = 1$, see \Cref{Tab: hash empirical quotient space dimension}.
    \clearpage 
    \begin{table}[H]
        \centering
        \caption{Empirical growth of the vector space dimension of the quotient space of \ArionHash.}
        \label{Tab: hash empirical quotient space dimension}
        \begin{tabular}{ M{5mm} | M{5mm} | M{7mm} | M{1.6cm} }
            \toprule
            $n$ & $d_1$ & $d_2$ & $\dim_{\Fp, \text{emp}}$ \\
            \midrule

            $2$ & $3$ & $7$   & $35$    \\
            $2$ & $3$ & $257$ & $1285$  \\
            $3$ & $3$ & $7$   & $91$    \\
            $3$ & $3$ & $257$ & $3341$  \\
            $4$ & $3$ & $7$   & $203$   \\
            $5$ & $3$ & $7$   & $427$   \\

            $2$ & $5$ & $7$   & $49$    \\
            $2$ & $5$ & $257$ & $1799$  \\
            $3$ & $5$ & $7$   & $133$   \\
            $3$ & $5$ & $257$ & $4833$  \\
            $4$ & $5$ & $7$   & $301$   \\
            $5$ & $5$ & $7$   & $637$   \\

            \bottomrule
        \end{tabular}
    \end{table}
    From our experiments we conjecture that the quotient space dimension grows or is bounded via
    \begin{equation}\label{Equ: hash quotient space dimension}
        \dim_{\Fp} \left( \mathcal{F}_\ArionHash \right) (n, r, d_1, d_2) = \Big( 2^{n - 1} \cdot d_2 \cdot \left( d_1 + 1 \right) - d_1 \cdot d_2 \Big)^r.
    \end{equation}
    As for \Arion, we apply this equation in \Cref{Equ: complexity deterministic solving,Equ: complexity probabilistic solving} to estimate the cost of solving the \ArionHash polynomial system.

    Since the polynomial degrees of \ArionHash coincide with the ones for \Arion also the Gr\"obner basis complexity estimates of \ArionHash and \Arion coincide.
    \clearpage 
    \begin{table}[H]
        \centering
        \caption{Empirical cost estimation of preimage Gr\"obner basis attacks on \ArionHash for primes $p \geq 2^{60}$ and $d_2 \geq 121$.
            The column GB min contains the complexity of Gr\"obner basis computations estimated with the highest polynomial degree in the \Arion polynomial system.
            The column GB MB contains the complexity Gr\"obner basis computations estimated via the Macaulay bound.
            We assume that the adversary has an optimal Gaussian elimination algorithm with $\omega = 2$.}
        \label{Tab: hash Groebner basis complexity}
        \resizebox{0.7\textwidth}{!}{
            \begin{tabular}{ M{5mm} | M{5mm} | M{15mm} | M{15mm} | M{25mm} | M{25mm} }
                \toprule
                $n$ & $r$ & GB min (bits) & GB MB (bits) & Deterministic Solving (bits) & Probabilistic Solving (bits) \\
                \midrule

                \multicolumn{6}{c}{$d_1 = 3$} \\
                \midrule

                $3$ & $5$ & $143$ & $150$ & $158$ & $110$ \\
                $3$ & $6$ & $162$ & $170$ & $190$ & $132$ \\

                $4$ & $4$ & $143$ & $160$ & $141$ & $98$  \\
                $4$ & $6$ & $187$ & $210$ & $212$ & $146$ \\

                $5$ & $4$ & $162$ & $201$ & $154$ & $107$ \\
                $5$ & $5$ & $187$ & $235$ & $193$ & $133$ \\

                $6$ & $4$ & $181$ & $257$ & $167$ & $115$ \\
                $6$ & $5$ & $208$ & $301$ & $208$ & $143$ \\

                $8$ & $3$ & $279$ & $338$ & $143$ & $100$ \\
                $8$ & $4$ & $345$ & $423$ & $191$ & $132$ \\

                \midrule
                \multicolumn{6}{c}{$d_1 = 5$} \\
                \midrule

                $3$ & $4$ & $123$ & $131$ & $133$ & $93$  \\
                $3$ & $6$ & $162$ & $173$ & $200$ & $138$ \\

                $4$ & $4$ & $143$ & $143$ & $167$ & $103$ \\
                $4$ & $5$ & $179$ & $166$ & $195$ & $128$ \\

                $5$ & $3$ & $162$ & $215$ & $161$ & $111$ \\
                $5$ & $5$ & $187$ & $251$ & $201$ & $139$ \\

                $6$ & $3$ & $172$ & $225$ & $130$ & $91$  \\
                $6$ & $5$ & $244$ & $328$ & $217$ & $149$ \\

                $8$ & $3$ & $309$ & $366$ & $148$ & $103$ \\
                $8$ & $4$ & $385$ & $461$ & $198$ & $137$ \\

                \bottomrule
            \end{tabular}
        }
    \end{table}

    An adversary can also set up a collision attack via the equation
    \begin{equation}
        \ArionHash (\mathbf{x}) = \ArionHash (\mathbf{y}),
    \end{equation}
    where $\mathbf{x}, \mathbf{y} \in \Fp^\ast$ are variables.
    We can construct the corresponding collision polynomial system by first preparing two preimage polynomial systems $\mathcal{F}_1$ and $\mathcal{F}_2$, then we connect the last round of both systems via
    \begin{align}
        \begin{pmatrix}
            0 & 0 & \ldots & 0     \\
            n & 1 & \ldots & n - 1 \\
            &   & \ddots &       \\
            2 & 3 & \ldots & 1
        \end{pmatrix}
        \tilde{\mathcal{F}} \left( \mathbf{\hat{x}}^{(r - 1)} \right) + \mathbf{c}_r
        &=
        \begin{pmatrix}
            0 \\
            \mathbf{x}_{out}
        \end{pmatrix}
        , \\
        x_n^{(r - 1)} - \left( z^{(r)} \right)^{d_2} &= 0, \\
        \begin{pmatrix}
            0 & 0 & \ldots & 0     \\
            n & 1 & \ldots & n - 1 \\
            &   & \ddots &       \\
            2 & 3 & \ldots & 1
        \end{pmatrix}
        \tilde{\mathcal{F}} \left( \mathbf{\hat{y}}^{(r - 1)} \right) + \mathbf{c}_r
        &=
        \begin{pmatrix}
            0 \\
            \mathbf{y}_{out}
        \end{pmatrix}
        , \\
        y_n^{(r - 1)} - \left( z^{(r)} \right)^{d_2} &= 0, \\
        \begin{pmatrix}
            1 & 2 & \ldots & n \\
            0 & 0 & \ldots & 0 \\
            &   & \ddots &   \\
            0 & 0 & \ldots & 0
        \end{pmatrix}
        \left( \tilde{\mathcal{F}} \left( \mathbf{\hat{x}}^{(r - 1)} \right) - \tilde{\mathcal{F}} \left( \mathbf{\hat{y}}^{(r - 1)} \right) \right) &= \mathbf{0}.
    \end{align}
    Note that the collision polynomial system consists of $2 \cdot r \cdot \left( n + 1 \right) - 1$ many equations in $2 \cdot r \cdot \left( n + 1 \right)$ many variables, hence it is not fully determined.
    Therefore, one has to guess at least one variable in the collision polynomial system.
    In all our experiments we guessed one value of $\mathbf{y}_{out}$ and then computed the DRL Gr\"obner basis of the system.
    We observed that the position of our guess in $\mathbf{y}_{out}$ did never affect the dimension of the quotient space dimension.
    In \Cref{Tab: hash empirical collision quotient space dimension} we record the observed quotient space dimension in our experiments.
    \begin{table}[H]
        \centering
        \caption{Empirical growth of the vector space dimension of the quotient space of the collision polynomial system of \ArionHash.}
        \label{Tab: hash empirical collision quotient space dimension}
        \begin{tabular}{ M{5mm} | M{5mm} | M{5mm} | M{5mm} | M{1.6cm} }
            \toprule
            $n$ & $r$ & $d_1$ & $d_2$ & $\dim_{\Fp, \text{emp}}$ \\
            \midrule

            $2$ & $1$ & $3$ & $3$ & $225$   \\

            $2$ & $2$ & $3$ & $3$ & $50625$ \\

            $2$ & $1$ & $5$ & $5$ & $1225$  \\

            $2$ & $1$ & $3$ & $7$ & $1225$  \\

            $3$ & $1$ & $3$ & $7$ & $8281$  \\

            $2$ & $1$ & $5$ & $7$ & $2401$  \\

            $3$ & $1$ & $3$ & $3$ & $1521$  \\

            $3$ & $1$ & $5$ & $5$ & $9025$  \\

            $3$ & $1$ & $3$ & $7$ & $8281$  \\

            $3$ & $1$ & $5$ & $7$ & $17689$ \\

            $4$ & $1$ & $3$ & $3$ & $7569$  \\

            $4$ & $1$ & $3$ & $7$ & $41209$ \\

            \bottomrule
        \end{tabular}
    \end{table}

    From our experiments we conjecture that the quotient space dimension grows or is bounded via
    \begin{equation}\label{Equ: collision quotient space dimension}
        \dim_{\Fp} \left( \mathcal{F}_{\ArionHash, coll} \right) \left( n, r, d_1, d_2 \right) = \left( \dim_{\Fp} \left( \mathcal{F}_\ArionHash \right) \left( n, r, d_1, d_2 \right) \right)^2
    \end{equation}
    Hence, the complexity estimate for the collision attack corresponds to the one for the preimage attack with doubled number of rounds.
    In particular, \Cref{Tab: hash Groebner basis complexity} implies at least $128$ bit security against a collision Gr\"obner basis attack.

    \section{Performance Evaluation}
    \subsection{\libsnark Implementation}\label{Sec: libsnark performance}
    We implemented \ArionHash, \Griffin\ and \Poseidon using the \verb!C++! library \libsnark \cite{libsnark} that is used in the privacy-protecting digital currency
    Zcash \cite{Zcash}.
    The comparative result of our experiment is given in \Cref{Tab: libsnark runtimes}. All experiments were run on a system with an Intel Core i7-11800H CPU and 32 GB RAM on a Clear Linux instance, using the \texttt{g++-12} compiler with \texttt{-O3 -march=native} flags.
    Our result shows that \ArionHash significantly outperforms \Poseidon showing 2x efficiency improvement, and the \aArionHash\ is considerably faster than \Griffin for practical parameter choices e.g. $n = 3$ or $4$ in Merkle tree hashing mode.
    \begin{table}[H]
        \centering
        \caption{Performance of various hash functions for proving the membership of a Merkle tree accumulator over $\text{BN}254$ with $d_2 = 257$.
            Proving times are in ms.
            For all hash functions verification requires less than $\SI{20}{\milli\second}$.}
        \label{Tab: libsnark runtimes}
        \resizebox{1.0\textwidth}{!}{
            \begin{tabular}{ M{1.1cm} | M{1.1cm} M{1.1cm} M{1.1cm} | M{1.1cm} M{1.1cm} M{1.1cm} | M{1.1cm} M{1.1cm} M{1.1cm} | M{1.1cm} M{1.1cm} M{1.1cm} }
                \toprule
                & \multicolumn{12}{ c  }{Time (ms)} \\
                \midrule

                Height & \multicolumn{3}{ c | }{\ArionHash} & \multicolumn{3}{ c | }{\aArionHash} & \multicolumn{3}{ c | }{\Griffin} & \multicolumn{3}{ c }{\Poseidon} \\
                \midrule
                $d = 5$ & $n = 3$ & $n = 4$ & $n = 8$ & $n = 3$ & $n = 4$ & $n = 8$ & $n = 3$  & $n = 4$  & $n = 8$ & $n = 3$  & $n = 4$  & $n = 8$ \\
                \midrule
                $4$   & $101$ & $103$ & $142$  & $73$  & $87$  & $143$  & $88$  & $99$  & $133$ & $186$  & $212$  & $274$  \\
                $8$   & $211$ & $216$ & $294$  & $145$ & $177$ & $294$  & $181$ & $209$ & $270$ & $386$  & $417$  & $566$  \\
                $16$  & $392$ & $401$ & $554$  & $278$ & $334$ & $553$  & $338$ & $387$ & $505$ & $745$  & $805$  & $1095$ \\
                $32$  & $730$ & $751$ & $1046$ & $509$ & $646$ & $1047$ & $622$ & $727$ & $980$ & $1422$ & $1550$ & $2111$ \\
                \bottomrule
            \end{tabular}
        }
    \end{table}

    \subsection{Dusk Network \Plonk}\label{Sec: Dusk Plonk performance}
    We implemented \ArionHash in \verb|Rust| using the Dusk Network \Plonk \cite{Dusk-Plonk} library which uses 3-wire constraints, and compared its performance to the Dusk Network \Poseidon \cite{Dusk-Poseidon} reference implementation.
    We note that for \ArionHash we only implemented a generic affine layer, i.e., \ArionHash needs $n \cdot (n - 1)$ constraints for matrix multiplication and $n$ constraints for constant addition.
    The comparative result of our experiment is given in \Cref{Tab: plonk runtimes}.
    All experiments were run on a system with a Intel Core i7-10700 and 64 GB RAM.
    Our result shows that \ArionHash outperforms \Poseidon showing 2x efficiency improvement in Merkle tree hashing mode.
    \clearpage 
    \begin{table}[H]
        \centering
        \caption{Performance of various hash functions for proving the membership of a Merkle tree accumulator over $\text{BLS}12$ for $d_1 = 5$ and $n = 5$.
            For each hash function and each Merkle tree height $1000$ samples were collected.
            Average times plus-minus the standard deviation are given in the table.
            All timings are in ms.
            For all hash functions verification requires less than $\SI{10}{\milli\second}$.}
        \label{Tab: plonk runtimes}
        \resizebox{1.0\textwidth}{!}{
            \begin{tabular}{ M{3cm} | M{1.6cm} M{1.6cm} M{1.6cm} M{1.6cm} | M{1.6cm} M{1.6cm} M{1.6cm} M{1.6cm} }
                \toprule
                \multicolumn{9}{ c  }{Time (ms)} \\
                \midrule

                & \multicolumn{4}{ c | }{\ArionHash} & \multicolumn{4}{ c }{\Poseidon} \\
                \midrule

                Native sponge    & \multicolumn{4}{ c | }{$0.121 \pm 0.001$} & \multicolumn{4}{ c }{$0.101 \pm 0.001$} \\
                Proof generation & \multicolumn{4}{ c | }{$167 \pm 1$}       & \multicolumn{4}{ c }{$313 \pm 4$} \\

                \midrule
                \multicolumn{9}{ c }{Merkle Tree} \\
                \midrule
                Height & $4$ & $8$ & $16$ & $32$ & $4$ & $8$ & $16$ & $32$ \\
                \midrule

                Proof generation & $323 \pm 4$ & $533 \pm 5$ & $1011 \pm 13$  & $2034 \pm 16$ & $524 \pm 6$ & $1008 \pm 14$ & $1975 \pm 27$ & $4015 \pm 57$ \\
                \bottomrule
            \end{tabular}
        }
    \end{table}
\end{appendix}

\end{document}